\documentclass[a4paper,english]{lipics-v2019}
\pdfoutput=1

\usepackage{microtype}
\usepackage{amssymb}
\usepackage{mathtools}
\usepackage{nicefrac}
\usepackage[utf8]{inputenc}
\usepackage[linesnumbered, vlined]{algorithm2e}

\DeclareGraphicsExtensions{.pdf}

\newcommand{\Oh}{O}
\newcommand{\OhOp}[1]{O\mathopen{}\mathclose\bgroup\left( #1 \aftergroup\egroup\right)}

\usepackage{xspace}
\newcommand{\FPT}{{\sf FPT}\xspace}
\newcommand{\NP}{{\sf NP}\xspace}

\usepackage{tabularx}
\newcommand{\prob}[3]{
\begin{center}
\begin{tabularx}{\textwidth}{lX}
	\multicolumn{2}{l}{#1}\\
	{\bf Input:}&{#2}\\
	{\bf Question:}&{#3}
\end{tabularx}
\end{center}
}

\newcommand{\hy}{\hbox{-}\nobreak\hskip0pt}

\DeclareMathOperator{\argmin}{\operatorname{argmin}}

\def \td {{\rm td}}

\def\ve#1{\mathchoice{\mbox{\boldmath$\displaystyle\bf#1$}}
{\mbox{\boldmath$\textstyle\bf#1$}}
{\mbox{\boldmath$\scriptstyle\bf#1$}}
{\mbox{\boldmath$\scriptscriptstyle\bf#1$}}}

\newcommand\veb{{\ve b}}

\newcommand\ved{{\ve d}}

\newcommand\veg{{\ve g}}
\newcommand\veh{{\ve h}}
\newcommand\vel{{\ve l}}
\newcommand\ven{{\ve n}}

\newcommand\vep{{\ve p}}

\newcommand\veu{{\ve u}}

\newcommand\vew{{\ve w}}
\newcommand\vex{{\ve x}}
\newcommand\vey{{\ve y}}
\newcommand\vez{{\ve z}}

\def\R{\mathbb{R}}
\def\Z{\mathbb{Z}}
\def\N{\mathbb{N}}
\def\G{\mathcal{G}}

\newcommand\gc{\mathtt{g_1}}

\makeatletter
\newtheorem*{rep@theorem}{\rep@title}
\newcommand{\newreptheorem}[2]{%
\newenvironment{rep#1}[1]{%
 \def\rep@title{#2 \ref{##1}}%
 \begin{rep@theorem}}%
 {\end{rep@theorem}}}
\makeatother

\newreptheorem{theorem}{Theorem}

\title{Evaluating and Tuning $n$-fold Integer Programming}
\titlerunning{Evaluating and Tuning $n$-fold Integer Programming}
\author{Kateřina Altmanová}
{Department of Applied Mathematics, Charles University, Prague, Czech Republic}
{kacka@kam.mff.cuni.cz}
{}{Author was supported the project 17-09142S of GA~ČR.}
\author{Dušan Knop}
{
Algorithmics and Computational Complexity, Faculty~IV, TU Berlin
\and
Department of Theoretical Computer Science, Faculty of Information Technology,\\ Czech Technical University in Prague, Prague, Czech Republic
}
{dusan.knop@fit.cvut.cz}
{0000-0003-2588-5709}{Author supported by the project P202/12/G061 of GA~ČR.}
\author{Martin Koutecký}
{
Faculty of Industrial Engineering and Management, Technion -- Israel Institute of Technology\\{Haifa, Israel}
\and
Computer Science Institute of Charles University, Charles University, Prague, Czech Republic
}
{koutecky@technion.ac.il}
{0000-0002-7846-0053}{Author supported by a postdoctoral fellowship at the Technion funded by the Israel Science Foundation grant 308/18, by the project 17-09142S of GA~ČR, and by Charles University project UNCE/SCI/004.}

\authorrunning{K. Altmanová, D. Knop, and M. Koutecký} 

\Copyright{Kateřina Altmanová, Dušan Knop, and Martin Koutecký}



\category{}

\relatedversion{}

\supplement{\url{https://github.com/katealtmanova/nfoldexperiment}}

\funding{}

\acknowledgements{}

\EventEditors{John Q. Open and Joan R. Access}
\EventNoEds{2}
\EventLongTitle{42nd Conference on Very Important Topics (CVIT 2016)}
\EventShortTitle{CVIT 2016}
\EventAcronym{CVIT}
\EventYear{2016}
\EventDate{December 24--27, 2016}
\EventLocation{Little Whinging, United Kingdom}
\EventLogo{}
\SeriesVolume{42}
\ArticleNo{23}
\nolinenumbers 
\hideLIPIcs  

\begin{document}

\maketitle

\begin{abstract}
  In recent years, algorithmic breakthroughs in stringology, computational social choice, scheduling, etc., were achieved by applying the theory of so-called $n$-fold integer programming.
  An $n$-fold integer program (IP) has a highly uniform block structured constraint matrix.
  Hemmecke, Onn, and Romanchuk [Math. Programming, 2013] showed an algorithm with runtime $\Delta^{\Oh(rst + r^2s)} n^3$, where $\Delta$ is the largest coefficient, $r,s$, and $t$ are dimensions of blocks of the constraint matrix and $n$ is the total dimension of the IP; thus, an algorithm efficient if the blocks are of small size and with small coefficients.
  The algorithm works by iteratively improving a feasible solution with augmenting steps, and $n$-fold IPs have the special property that augmenting steps are guaranteed to exist in a not-too-large neighborhood.
  However, this algorithm has never been implemented and evaluated.

  We have implemented the algorithm and learned the following along the way.
  The original algorithm is practically unusable, but we discover a series of improvements which make its evaluation possible.
  Crucially, we observe that a certain constant in the algorithm can be treated as a tuning parameter, which yields an efficient heuristic (essentially searching in a smaller-than-guaranteed neighborhood).
  Furthermore, the algorithm uses an overly expensive strategy to find a ``best'' step, while finding only an ``approximately best'' step is much cheaper, yet sufficient for quick convergence.
  Using this insight, we improve the asymptotic dependence on $n$ from $n^3$ to $n^2 \log n$.

  Finally, we tested the behavior of the algorithm with various values of the tuning parameter and different strategies of finding improving steps.
  First, we show that decreasing the tuning parameter initially leads to an increased number of iterations needed for convergence and eventually to getting stuck in local optima, as expected.
  However, surprisingly small values of the parameter already exhibit good behavior while significantly lowering the time the algorithm spends per single iteration.
  Second, our new strategy for finding ``approximately best'' steps wildly outperforms the original construction.
\end{abstract}

\section{Introduction}
\sloppy
In this article we consider the general integer linear programming (ILP) problem in standard form,
\begin{equation} \label{IP}
  \min\left\{\vew \vex \, \mid A\vex=\veb\,,\ \vel\leq\vex\leq\veu\,,\ \vex\in\Z^{n}\right\}. \tag{ILP}
\end{equation}
with $A$ an integer $m\times n$ matrix, $\veb\in\Z^m$, $\vew\in\Z^n$, $\vel,\veu\in(\Z\cup\{\pm\infty\})^n$.
It is well known to be strongly \NP-hard, but models many important problems in combinatorial optimization such as planning~\cite{planning}, scheduling~\cite{scheduling}, and transportation~\cite{transportation} and thus powerful generic solvers have been developed for it~\cite{Lodi:2010}.
Still, theory is motivated to search for tractable special cases.
One such special case is when the constraint matrix $A$ has a so-called $N$-fold structure:
\begin{align*}
A = E^{(N)} =
\left(
\begin{array}{cccc}
E_1    & E_1    & \cdots & E_1    \\
E_2    & 0      & \cdots & 0      \\
0      & E_2    & \cdots & 0      \\
\vdots & \vdots & \ddots & \vdots \\
0      & 0      & \cdots & E_2    \\
\end{array}
\right) \enspace .\label{nfold}
\end{align*}
Here, $r,s,t,N \in \N$, $\veu, \vel, \vew \in \Z^{Nt}$, $\veb \in \Z^{r+Ns}$,
$E^{(N)}$ is an $(r+Ns)\times Nt$-matrix, $E_1 \in \Z^{r \times t}$ is an $r\times t$-matrix and $E_2 \in \Z^{s \times t}$ is an $s\times t$-matrix.
We call $E^{(N)}$ the \emph{$N$-fold product of $E = \left(\begin{smallmatrix}E_1\\E_2\end{smallmatrix}\right)$} and denote by $L$ the length of the binary encoding of the instance $(A,\vew, \veb, \vel, \veu)$\footnote{For clarity of exposition we shall no longer consider infinite lower and upper bounds. We note that this is without loss of generality by standard arguments: an instance with some bounds $\pm\infty$ is either unbounded or one may, in polynomial time, replace $\vel, \veu$ with auxiliary bounds $\vel', \veu'$ which are of polynomial length and do not change the optimal value of the instance.}.
Problem~\eqref{IP} with $A = E^{(N)}$ is known as \emph{$N$\hy fold integer programming} ($N$-fold IP).
Hemmecke, Onn, and Romanchuk~\cite{HemmeckeOR13} prove the following.
\begin{proposition}[{\cite[Theorem 6.2]{HemmeckeOR13}}]\label{thm:nfold}
There is an algorithm that solves\footnote{Given an IP, to \emph{solve} it means to either (i) declare it infeasible or unbounded or (ii) find its minimizer.}~\eqref{IP} with $A=E^{(N)}$ encoded with $L$ bits in time \mbox{$\Delta^{O(trs + t^2s)}\cdot n^3L$}, where \mbox{$\Delta = 1 + \max\{\|E_1\|_\infty, \|E_2\|_\infty\}$}.
\end{proposition}
Recently, algorithmic breakthroughs in stringology~\cite{KnopKM:2017esa}, computational social choice~\cite{KnopKM:2017stacs}, scheduling~\cite{ChenMarx:2018,JansenKMR:2018,KnopKoutecky:2017}, etc., were achieved by applying this algorithm and its subsequent non-trivial improvements.

The algorithm belongs to the larger family of augmentation (primal) algorithms.
It starts with an initial feasible solution $\vex_0 \in \Z^{Nt}$ and produces a sequence of increasingly better solutions $\vex_1, \dots, \vex_\sigma$ (better means $\vew \vex_\sigma < \vew \vex_{\sigma-1} < \cdots < \vew \vex_0$).
It is guaranteed that the algorithm terminates, that $\vex_\sigma$ is an optimal solution, and that the algorithm converges quickly, i.e., $\sigma$ is polynomial in the length of the input.
A key property of $N$-fold IPs is that, if an augmenting step exists, then it can be decomposed into a bounded number of elements of the so-called \emph{Graver basis} of $A$, which we denote $\G(A)$.
This in turn makes it possible to compute it using dynamic programming~\cite[Lemma 3.1]{HemmeckeOR13}.
In a sense, this property makes the algorithm a local search algorithm which is always guaranteed to find an improvement in a not-too-large neighborhood.
The bound on the number of elements or the size of the neighborhood which needs to be searched is called the \emph{Graver complexity of $E$}, denoted $g(E)$.
This, in turn, implies that, if an augmenting step exists, then there is always one with small $\ell_1$-norm; for a matrix $A$, we denote this bound $g_1(A) = \max_{\veg \in \G(A)} \|\veg\|_1$~\cite[Theorem 4]{PSP}.
However, the algorithm has never been implemented and evaluated.

\subsection{Our Contributions}
We have implemented the algorithm and tested it on two problems for which $N$-fold formulations were known: makespan minimization on uniformly related machines ($Q || C_{\max}$) and \textsc{Closest String}; we have used randomly generated instances.
The solver, tools, and e.g. many more plots can be accessed in a publicly accessible repository at \url{https://github.com/katealtmanova/nfoldexperiment}

In the course of implementing the algorithm we learn the following.
The algorithm in its initial form is practically unusable due to an \emph{a priori} construction of the Graver basis $\G(E_2)$ of size exponential in $s,t$ and $\Delta$, and a related (even larger) set $Z(E)$, whose size is exponential in $r,s,t$ and $\Delta$.
However, we discover a series of improvements (some building on recent insights~\cite{PSP}) which avoid the construction of these two sets.
Moreover, we adjust the algorithm to treat $g_1(A)$ as a tuning parameter $\gc$, which turns it into a heuristic (i.e., an optimal solution or polynomial runtime is not guaranteed; we shall discuss this topic in more detail later).

We also study the \emph{augmentation strategy}, which is the way the algorithm chooses an augmenting step among all the possible options.
The original algorithm uses an overly expensive strategy to find a ``best'' step,
which means that a large number of possible steps is evaluated in each iteration.
We show that finding only an ``approximately best'' step is sufficient to obtain asymptotically equivalent convergence rate, and the work per iteration decreases exponentially.
Using this insight, we improve the asymptotic dependence on $N$ from $N^3$ to $N^2 \log N$.
Together with recent improvements, this yields the currently asymptotically fastest algorithm for $N$-fold IP:
\begin{theorem}\label{thm:gammanfold}
Problem~\eqref{IP} with $A = E^{(N)}$ can be solved in time $\Delta^{r^2s + rs^2} (Nt)^2 \log (Nt) M$, where $M = \log(\vew\vex^* - \vew\vex_0)$ for some minimizer $\vex^*$ of $\vew \vex$.
\end{theorem}

Finally, we evaluate the behavior of the algorithm.
We ask how is the performance of the algorithm (in terms of number of dynamic programming calls and quality of the returned solution) influenced by
\begin{enumerate}
\item the choice of the tuning parameter $1 < \gc \leq g_1(A)$?
\item the choice of the augmentation strategy between ``best step'', ``approximate best step'', and ``any step''?
\end{enumerate}
As expected, with $\gc$ moving from $g_1(A)$ to $1$, we first see an increase in the number of iterations needed for convergence and eventually the algorithm gets stuck in a local optima.
However, surprisingly small values (e.g. $\gc=50$ when $g_1(A) > 10^{11}$) of the parameter already exhibit close to optimal behavior while significantly decreasing the time spend per iteration.
Second, our new strategy for finding ``approximately best'' steps outperforms the original construction by orders of magnitude, while the naive ``any step'' strategy behaves erratically.

We note that at this stage we are \emph{not} (yet) interested in showing supremacy over existing algorithms; we simply want to understand the practical behavior of an algorithm whose theoretical importance was recently highlighted.
For this reason our experimental focus is on the two aforementioned questions rather than simply measuring the time.
Unfortunately, our data does not indicate any slowdown of a commercial MILP solver based on the number of bricks, which is required to give the algorithm of Theorem~\ref{thm:gammanfold} a chance to beat it.

Due to the rigid format of $E^{(N)}$ we are limited to few problems for which $N$-fold formulations are known.
Regarding instances, for \textsc{Closest String} we use the same approach as Chimani et al.~\cite{ClosestString}; for \textsc{Makespan Minimization} we generate our own data because standard benchmarks are not limited to short jobs or few types of jobs.

\subsection{Related Work}
Our work mainly relates to \emph{primal heuristics}~\cite{primalheu} for MIPs which are used to help reach optimality faster and provide good feasible solutions early in the termination process.
Specifically, our algorithm is a \emph{neighborhood} (or \emph{local}) \emph{search algorithm}.
The standard paradigm is \emph{Large Neighborhood Search} (LNS)~\cite{lns} with specializations such as for example \emph{Relaxation Induced Neighborhood Search} (RINS)~\cite{rins} and \emph{Feasibility Pump}~\cite{feasibilitypump}.
In terms of this paradigm, our proposed algorithm searches in the neighborhood induced by the $\ell_1$-distance around the current feasible solution and the search procedure is formulated as an ILP subproblem with the additional constraint $\|\vex\|_1 \leq \gc$.
In this sense the closest technique to ours is \emph{local branching}~\cite{localbranching} which also searches in the $\ell_1$-neighborhood; however, we treat the discovered step as a \emph{direction} and apply it exhaustively, so, unlike in local branching, we make long steps.
Moreover, local branching was mainly applied to binary ILPs without any additional structure of the constraint matrix.

On the theoretical side, very recently Koutecký et al.~\cite{PSP} have studied parameterized strongly polynomial algorithms for various block-structured ILPs, not just $N$-fold IP.
Eisenbrand et al.~\cite{EisenbrandHK:2018} independently (and using slightly different techniques) arrive at the same complexity of $N$-fold IP  as our Theorem~\ref{thm:gammanfold}.
Jansen et al.~\cite{JansenLR:2018} have shown a near-linear time algorithm for $N$-fold IP with linear objectives.
Their approach is relevant to implementations of an \FPT algorithm for $N$-fold IP, however due to our approach of using existing ILP solvers as a subroutine we do not exploit it.


\section{Preliminaries}
\label{sec:preliminaries}
For positive integers $m,n$ we set $[m,n] = \{m,\ldots, n\}$ and $[n] = [1,n]$.
We write vectors in boldface (e.g., $\vex, \vey$) and their entries in normal font (e.g., the $i$-th entry of~$\vex$ is~$x_i$).
Given the problem~\eqref{IP}, we say that $\vex$ is \emph{feasible} for~\eqref{IP} if $A\vex = \veb$ and $\vel \leq \vex \leq \veu$.

\subsection{Graver bases and augmentation.}
Let us now introduce Graver bases and discuss how they can be used for optimization.
We also recall $N$-fold IPs; for background, we refer to the books of Onn~\cite{Onn2010} and De Loera et al.~\cite{DeLoeraEtAl2013}.

\subparagraph*{$N$-fold IP}
The structure of $E^{(N)}$ allows us to divide the $Nt$ variables of $\vex$ into $N$ \textit{bricks} of size~$t$.
We use subscripts to index within a brick and superscripts to denote the index of the brick, i.e.,~$x_j^i$ is the $j$-th variable of the $i$-th brick with $j \in [t]$ and $i \in [N]$.

Let $\vex, \vey$ be $n$-dimensional integer vectors.
We call $\vex, \vey$ \emph{sign\hy{}compatible} if they lie in the same orthant, that is, if for each $i \in [n]$ it holds that $x_i \cdot y_i \geq 0$.
We call $\sum_i \veg^i$ a \emph{sign\hy{}compatible sum} if all $\veg^i$ are pair-wise sign\hy{}compatible.
Moreover, we write $\vey \sqsubseteq \vex$ if $\vex$ and $\vey$ are sign\hy{}compatible and $|y_i| \leq |x_i|$ for each $i \in [n]$.
Clearly, $\sqsubseteq$ imposes a partial order called ``conformal order'' on $n$-dimensional vectors.
For an integer matrix $A \in \Z^{m \times n}$, its \emph{Graver basis} $\G(A)$ is the set of $\sqsubseteq$-minimal non-zero elements of the \emph{lattice} of $A$, $\ker_{\Z}(A) = \{\vez \in \Z^n \mid A \vez = \mathbf{0}\}$.
An important property of $\G(A)$ is the following.
\begin{proposition}[{\cite[Lemma 3.4]{Onn2010}}]
  \label{prop:graver_conformal_sum}
  Every integer vector $\vex \neq \mathbf{0}$ with $A \vex = \mathbf{0}$ is a sign\hy{}compatible sum $\vex = \sum_{i=1}^{n'} \alpha_i \veg^i$, $\alpha_i \in \N$, $\veg^i \in \G(A)$ and $n' \leq 2n-2$.
\end{proposition}

Let $\vex$ be a feasible solution to~\eqref{IP}.
We call $\veg$ an \emph{$\vex$-feasible step} (or simply \emph{feasible step} if $\vex$ is clear) if $\vex + \veg$ is feasible for~\eqref{IP}.
Further, we call a feasible step $\veg$ \emph{augmenting} if $\vew(\vex+\veg) < \vew\vex$; note that $\veg$ decreases the objective by $\vew \veg$.
An augmenting step $\veg$ and a \emph{step length} $\lambda \in \N$ form an \emph{$\vex$\hy{}feasible step pair} with respect to a feasible solution $\vex$ if $\vel \le \vex + \lambda\veg \le \veu$.
A pair $(\lambda, \veg) \in \left(\N \times \ker_{\Z}(A)\right)$ is a \emph{$\lambda$-Graver-best step pair} and $\lambda \veg$ is a \emph{$\lambda$-Graver-best step} if it is feasible and for every feasible step pair $(\lambda, \veg')$, $\veg' \in \G(A)$, we have $\vew \veg \leq \vew \veg'$.
An augmenting step $\veg$ and a step length $\lambda \in \N$ form a \emph{Graver\hy best step pair} if it is $\lambda$-Graver-best and it minimizes $\vew \lambda' \veg'$ over all $\lambda' \in \N$, where $(\lambda', \veg')$ is a $\lambda'$-Graver-best step pair.
We say that $\lambda \veg$ is a \emph{Graver-best step} if $(\lambda, \veg)$ is a Graver-best step pair.

The \emph{Graver\hy{}best augmentation procedure} for~\eqref{IP} with a given feasible solution $\vex_0$ and initial value $i=0$ works as follows:
\begin{enumerate}
  \item \label{step1}If there is no Graver\hy{}best step for $\vex_i$, return it as optimal.
  \item If a Graver\hy{}best step $\lambda \veg$ for $\vex_i$ exists, set $\vex_{i+1} := \vex_i + \lambda \veg$, $i:=i+1$, and go to \ref{step1}.
\end{enumerate}
\begin{proposition}[Convergence bound~{\cite[Lemma 3.10]{Onn2010}}]
\label{prop:graverbest}
  Given a feasible solution~$\vex_0$ for~\eqref{IP}, the Graver\hy{}best augmentation procedure finds an optimum in at most $3n \log M$ steps, where $M = \vew(\vex_0 - \vex^*)$ and $\vex^*$ is any minimizer of $\vew\vex$.
\end{proposition}

By standard techniques (detecting unboundedness etc.) we can ensure that $\log M \leq L$.

\section{Approximate Graver-best Steps}
In this section we introduce the notion of a $c$-approximate Graver\hy{}best step (Definition~\ref{def:apx_gb_step}), show that such steps exhibit good convergence (Lemma~\ref{lem:approxgraverbest}), can be easily obtained (Lemma~\ref{lem:log_gamma}), and result in a significant speed-up of the $N$-fold IP algorithm (Theorem~\ref{thm:gammanfold}).

\begin{definition}[$c$-approximate Graver-best step] \label{def:apx_gb_step}
Let $c \in \R$ with $c \geq 1$.
Given an instance of~\eqref{IP} and a feasible solution $\vex$, we say that an $x$-feasible step $\veh$ is a \emph{$c$\hy{}approximate Graver-best step for $\vex$} if, for every $\vex$-feasible step pair $(\lambda, \veg) \in \left(\N \times \G(A)\right)$, we have $\vew \veh \leq \frac{1}{c} \cdot \lambda \vew  \veg $.
\end{definition}

Recall the Graver\hy{}best augmentation procedure.
We call its analogue where we replace a Graver\hy{}best step with a $c$\hy{}approximate Graver\hy{}best step the \emph{$c$-approximate Graver\hy{}best augmentation procedure}.

\begin{lemma}[$c$-approximate convergence bound]
\label{lem:approxgraverbest}
Given a feasible solution~$\vex_0$ for~\eqref{IP}, the $c$\hy{}approximate Graver\hy{}best augmentation procedure finds an optimum of~\eqref{IP} in at most $c \cdot 3n \log M$ steps, where $M = \vew(\vex_0 - \vex^*)$ and $\vex^*$ is any minimizer of\/ $\vew \vex$.
\end{lemma}
\begin{proof}
The proof is a straightforward adaptation of the proof of Proposition~\ref{prop:graverbest} which we first repeat here for convenience.
Let $\vex^*$ be a minimizer and let $\veh = \vex^* - \vex_0$.
Since $A\veh = \mathbf{0}$, by Proposition~\ref{prop:graver_conformal_sum}, $\veh = \sum_{i=1}^{n'} \alpha_i \veg^i$ for some $n' \leq 2n-2$, $\alpha_i \in \N$, $\veg^i \in \G(A)$, $i \in [n']$.
Thus by an averaging argument, an $\vex$-feasible step pair $(\lambda, \veg)$ such that $\lambda \veg$ is a Graver\hy{}best step must satisfy $\vew \lambda \veg \leq \frac{1}{2n-2} M$.
In other words, any Graver\hy{}best step pair improves the objective function by at least a $\frac{1}{2n-2}$-fraction of the total optimality gap $M$, and thus $3n \log M$ steps suffice to reach an optimum (cf.~\cite[Lemma 3.10]{Onn2010}).

It is straightforward to see that a $c$-approximate Graver\hy{}best step satisfies $\vew\vex - \vew(\vex + \lambda \veg) \leq \frac{c}{2n-2} M$, and thus $c(3n) \log M$ steps suffice.
\end{proof}

\begin{lemma}[Powers of $c$ step lengths] \label{lem:log_gamma}
Let $c \in \N$, $\vex$ be a feasible solution of~\eqref{IP}, and let
\[
\Gamma_{c\text{-apx}} = \left\{c^i \mid \exists \veg \in \G(A): \, \vel \leq \vex + c^i \veg \leq \veu \right\} \,.
\]
Let $(\lambda, \veg) \in \left(\Gamma_{c\text{-apx}} \times \G(A)\right)$ be an $\vex$-feasible step pair such that $\lambda \veg \leq \lambda' \veg'$ for any $\vex$\hy{}feasible step pair $(\lambda', \veg') \in \left(\Gamma_{c\text{-apx}} \times \G(A)\right)$.
Then $\lambda \veg$ is a $c$-approximate Graver-best step.
\end{lemma}
\begin{proof}
Let $(\lambda, \veg)$ satisfy the assumptions, and let $(\tilde{\lambda}, \tilde{\veg}) \in \left(\N \times \G(A)\right)$ be a Graver-best step pair.
Let $\lambda'$ be a nearest smaller power of $c$ from $\tilde{\lambda}$, and observe that $\lambda' \tilde{\veg}$ is a $c$-approximate Graver-best step because $\lambda' \geq \frac{\tilde{\lambda}}{c}$.
On the other hand, since $\lambda \veg$ is a $\lambda$-Graver\hy{}best step, we have $\lambda \veg \leq \lambda' \tilde{\veg}$ and thus $\lambda \veg$ is also a $c$-approximate Graver-best step, since we have $\vew\lambda\veg \le \vew\lambda'\tilde{\veg} \le \frac{1}{c}\vew\tilde{\lambda}\tilde{\veg}$.
\end{proof}

\begin{remark}
Lemma~\ref{lem:approxgraverbest} extends naturally to separable convex objectives; see the original proof~\cite[Lemma 3.10]{Onn2010}.
Moreover, Lemma~\ref{lem:log_gamma} also extends to separable convex objectives as was recently shown by Eisenbrand et al.~\cite{EisenbrandHK:2018}.
Thus Theorem~\ref{thm:gammanfold} (below) holds also for separable convex objectives.
\end{remark}

\begin{reptheorem}{thm:gammanfold}
Problem~\eqref{IP} with $A = E^{(N)}$ can be solved in time $\Delta^{\Oh(r^2s + rs^2)}(Nt)^2 \log (Nt) \cdot \log M$, where $M = \vew\vex^* - \vew\vex_0$ for some minimizer $\vex^*$ of $\vew \vex$.
\end{reptheorem}
\begin{proof}
Recall that $\Delta = \|A\|_\infty + 1$.
Koutecký et al.~\cite[Theorem 2]{PSP} show that a $\lambda$-Graver-best step can be found in time $\Delta^{\Oh(r^2s + rs^2)} Nt$.
Moreover, Hemmecke et al.~\cite{HemmeckeKW:2014} prove a proximity theorem which allows the reduction of an instance of~\eqref{IP} to an equivalent instance with new bounds $\vel', \veu'$ satisfying $\|\veu' - \vel'\|_\infty \leq Nt g_\infty$, with
\[
g_\infty = \max_{\veg \in \G(A)} \|\veg\|_\infty \leq \max_{\veg \in \G(A)} \|\veg\|_1 \leq (\Delta rs)^{\Oh(rs)} \,,
\]
where the last inequality can be found in the proof of~\cite[Theorem 4]{PSP}.
This bound implies that $\Gamma_{\text{2-apx}}$ from Lemma~\ref{lem:log_gamma} satisfies $\left| \Gamma_{\text{2-apx}} \right| \leq \log \|\veu'-\vel'\|_\infty \leq \log\left(Nt(\Delta r s)^{\Oh(rs)} \right) \leq \Oh(rs) \log (\Delta Ntrs)$.
By Lemma~\ref{lem:log_gamma}, finding a $\lambda$-Graver\hy{}best for each $\lambda \in \Gamma_{\text{2-apx}}$ and picking the minimum results in a $2$-approximate Graver\hy{}best step, and can be done in time $\Delta^{r^2s + rs^2} (Nt) \log (Nt)$.
By Lemma~\ref{lem:approxgraverbest}, $(4n-4) \log M$ steps suffice to reach the optimum.
\end{proof}

\section{Implementation}
We first give an overview of the original algorithm, which is our starting point.
Then we discuss our specific improvements and mention a few details of the software implementation.
\subsection{Overview of the Original Algorithm}
Recall that any $Nt$-dimensional vector related to $N$-fold IP is naturally partitioned into $N$ bricks of length $t$. In particular, this applies to the solution vector $\vex$ and any augmenting step $\veg$.
The key property of the $N$-fold product $E^{(N)}$ is that, regardless of $N \in \N$, the number of nonzero bricks of any $\veg \in \G(E^{(N)})$ is bounded by some constant $g(E)$ called the \emph{Graver complexity of~$E$}, and, moreover, that the sum of all non-zero bricks of $\veg$ can be decomposed into at most $g(E)$ elements of $\G(E_2)$~\cite[Lemma 3.1]{HemmeckeOR13}.
This facilitates the following construction.
Let
\[
Z(E) = \left\{\vez \in \Z^t\mid \exists \veg^1, \dots, \veg^k \in \G(E_2), \, k \leq g(E), \, \vez = \sum_{i=1}^k \veg^i  \right\} \enspace .
\]
Then, every prefix sum $\sum_{i=1}^j \veg^i$, $j \in [N]$, of the bricks of $\veg \in \G(E^{(N)})$ is contained in $Z(E)$ and a $\lambda$-Graver\hy{}best step, $\lambda \in \N$, can be found using dynamic programming over the elements of $Z(E)$.

To ensure that a Graver\hy{}best step is found, a set of step-lengths $\Gamma_{\text{best}}$ is constructed as follows.
Observe that any Graver\hy{}best (and thus feasible) step pair $(\lambda,\veg)\in \left(\N \times \G(E^{(N)}) \right)$, must satisfy that in at least one brick $i \in [N]$ it is ``tight'', that is, $(\lambda,\veg)$ is $\vex$-feasible while $(\lambda+1,\veg)$ is not specifically because $\vel^i \leq \vex^i + \lambda \veg^i \leq \veu^i$ holds but $\vel^i \leq \vex^i + (\lambda+1) \veg^i \leq \veu^i$ does not.
Thus, for each $\vez \in Z(E)$ and each $i \in [N]$, we find all the potentially ``tight'' step lengths $\lambda$ and add them to $\Gamma_{\text{best}}$, which results in a bound of $|\Gamma_{\text{best}}| \leq |Z(E)| \cdot N$.
Notice that this approach does not work for separable convex objectives for which a Graver\hy{}best step might not be tight in any coordinate.

For a overview of algorithm as described by Hemmecke, Onn, and Romanchuk see Algorithm~\ref{alg:HOR}.

\begin{algorithm}[bt]
\SetKwInOut{Input}{input}\SetKwInOut{Output}{output}
\Input{matrices $E_1,E_2$, positive integer $N$, and vectors $\veb,\vel,\veu,\vew$}
\Output{optimal solution to \eqref{IP} with $A = E^{(N)}$}

\SetKwProg{Fn}{Function}{}{}
\SetKwRepeat{Do}{do}{while}

\SetKwFunction{GraverComplexity}{GraverComplexity}
\SetKwFunction{FindFeasibleSolution}{FindFeasibleSolution}
\SetKwFunction{GraverBasis}{GraverBasis}
\SetKwFunction{BuildGammaBest}{BuildGammaBest}
\SetKwFunction{lambdaBestStep}{lambdaBestStep}
\SetKwFunction{DynamicProgramStates}{DynamicProgramStates}

$g = \GraverComplexity(E_1, E_2)$\;
$\vex_0 =$ \FindFeasibleSolution{$E,N,\veb,\vel,\veu$},\, $i = 0$\;
$\G(E_1) = \GraverBasis(E_1, g)$\;
$Z(E) = \DynamicProgramStates(\G(E_1), g)$\;

\Do{$\vex_{i-1} \neq \vex_i$}{
  $\Gamma_{\textrm{best}} =$ \BuildGammaBest{$\vex_i$}\;
  $i = i+1$\;
  \ForEach{$\lambda \in \Gamma$}{
    $\veg_\lambda =$ \lambdaBestStep{$Z(E), \lambda, \veg$}\;
  }
  $\vex_i = \vex_{i-1} + \argmin_{\left\{ \veg_\lambda \mid \lambda \in \Gamma \right\} } \vew\lambda \veg_\lambda$\;
}
\Return $\vex_i$\;

\caption{\label{alg:HOR}%
Pseudocode of the algorithm of Hemmecke, Onn, and Romanchuk.
}
\end{algorithm}

\subsection{Replacing Dynamic Programming with ILP} \label{subsec:dp_ilp}
We have started off by implementing the algorithm exactly as it is described by Hemmecke et al.~\cite{HemmeckeOR13}.
The first obstacle is encountered almost immediately and is contained in the constant $g(E)$.
This constant can be computed, but the computation is extremely difficult~\cite{FinholdHemmecke:2016,Hemmecke:2004}.
Another possibility is to estimate it, in which case it is almost always larger than $N$ and thus is essentially meaningless.
Finally, one can take the approach partially suggested in~\cite[Section 7]{HemmeckeOR13}, where we consider $g(E)$ in the construction of $Z(E)$ to be a tuning parameter and consider the approximate set $Z_{\mathtt{gc}}(E)$, $\mathtt{gc} \in \N$, obtained by taking sums of at most $\mathtt{gc}$ elements of $\G(E_2)$.
This makes the algorithm more practical, but turns it into a heuristic.

In spite of this sacrifice, already for small ($r=3$, $s=1$, $t=7$, $N=10$) instances and extremely small value of $\mathtt{gc} = 3$, the dynamic programming based on the $Z_{\mathtt{gc}}(E)$ construction was taking an unreasonably long time (over one minute).
Admittedly this could be improved; however, already for $\mathtt{gc} > 5$, it becomes infeasible to compute $Z_{\mathtt{gc}}(E)$, and for larger instances ($r > 5$, $t > 12$) it becomes very difficult to compute even $\G(E_2)$.
For these reasons we sought to completely replace the dynamic program involving $Z(E)$.

Koutecký et al.~\cite{PSP} show that all instances of~\eqref{IP} with the property that the so-called \emph{dual treedepth $\td_D(A)$ of $A$} is bounded and the largest coefficient $\|A\|_\infty$ is bounded also have the property that $g_1(A) = \max_{\veg \in \G(A)} \|\veg\|_1$ is bounded, which implies that augmenting steps can be found efficiently.
This class of ILPs contains $N$-fold IP.

The interpretation of the above fact is that, in order to solve~\eqref{IP}, it is sufficient to repeatedly (for different $\vex$ and $\lambda$) solve an auxiliary~\eqref{IP} instance
\begin{equation} \label{AugIP}
\min \left\{\vew \veh \mid A \veh = \mathbf{0},\, \vel \leq \vex + \lambda \veh \leq \veu, \, \|\veh\|_1 \leq g_1(A)\right\} \tag{AugILP}
\end{equation}
in order to find good augmenting steps; we note that the constraint $\|\veh\|_1 \leq g_1(A)$ can be linearized~\cite[Lemma 25]{PSP}.
The heuristic approach outlined above transfers easily: we replace $g_1(A)$ in~\eqref{AugIP} with some integer $\gc$, $1 < \gc  \leq g_1(A)$; this makes~\eqref{AugIP} easier to solve at the cost of losing the guarantee that an augmenting step is found if one exists.
In theory, solving~\eqref{AugIP} should be easier than solving the original instance~\eqref{IP} due to the special structure of $A$~\cite[Lemma 25]{PSP}.
Our approach here is to simply invoke an industrial MILP solver on~\eqref{AugIP} in order to find a $\lambda$-Graver\hy{}best step.

Note that the quantities $g(E)$ and $g_1(A)$ and the tuning parameters $\mathtt{gc}$ and $\gc$ are related but distinct.
First, $g(E)$ bounds the number of non-zero bricks of any element of $\G(A)$ and the number of elements of $\G(E_2)$ into which it decomposes, while $g_1(A)$ bounds the $\ell_1$-norm of any element of $\G(A)$.
It can be seen that bounded $g_1(A)$ implies bounded $g(E)$ and vice versa.
Second, $\mathtt{gc}$ and $\gc$ are tuning parameters derived from $g(E)$ and $g_1(A)$, respectively.
The crucial distinction is that the tuning parameter $\gc$ translates naturally into a linear constraint of~\eqref{AugIP} while $\mathtt{gc}$ only translates naturally to a construction of a restricted set of states $Z_{\mathtt{gc}}(E)$ which we are trying to avoid.

\subsection{Augmentation Strategy: Step Lengths}
\subsubsection*{Logarithmic $\Gamma$}
The majority of algorithms based on Graver basis augmentation rely on the Graver\hy{}best augmentation procedure~\cite{ChenMarx:2018,DeLoeraEtAl2013,HemmeckeOR13,KnopKM:2017esa,KnopKoutecky:2017,Onn2010}.
Consequently, these algorithms require finding (exact) Graver\hy{}best steps.
In the aforementioned algorithms this is always done using the construction of the set $\Gamma_{\text{best}}$ mentioned above, which is of size $f(k) \cdot n$ where $k$ is the relevant parameter (e.g., $(ars)^{\Oh(rst + st^2)}$ in the original algorithm for $N$-fold IP).
We replace this construction with $\Gamma_{\text{2-apx}} = \{1,2,4,8,\dots\}$ which, combined with the proximity technique, is only of size $\Oh(\log N)$ (Theorem~\ref{thm:gammanfold}); in particular, independent of the function $f(k)$.

\subsubsection*{Exhausting $\lambda$}
Moreover, we have noticed that sometimes the algorithm finds a step $\veg$ for $\lambda = 2^k$ which is not tight in any brick, and then repeatedly applies it for shorter step-lengths $\lambda' < \lambda$.
In other words, the discovered direction $\veg$ is not \emph{exhausted}.
Thus, for each $\lambda \in \N$, upon finding the $\lambda$-Graver\hy{}best step $\veg$, we replace $\lambda$ with the largest $\lambda' \geq \lambda$ for which $(\lambda', \veg)$ is still $\vex$-feasible.

\subsubsection*{Early termination}
Another observation is that in any given iteration of the algorithm, if $\lambda > 1$, then \emph{some} augmenting step has been found and if the computation is taking too long, we might terminate it and simply apply the best step found so far.

\subsubsection*{Initialize once}
We have noticed that a large portion of time spent on computing a $\lambda$\hy{}Graver\hy{}best step is taken by the initialization of the MILP model which is then solved very quickly.
However, notice that in the formulation of~\eqref{AugIP} the only changing parameters are the lower and upper bounds.
This leads us to a practical improvement: initialize the MILP model once in the beginning, and realize each~\eqref{AugIP} call by changing the bounds and reoptimizing the model.

For a overview of the newly proposed algorithm see Algorithm~\ref{alg:New}.

\begin{algorithm}[bt]
\SetKwInOut{Input}{input}\SetKwInOut{Output}{output}
\Input{matrices $E_1,E_2$, positive integers $N$, $c$ and $\gc$, and vectors $\veb,\vel,\veu,\vew$}
\Output{a feasible solution to \eqref{IP} with $A = E^{(N)}$}

\SetKwProg{Fn}{Function}{}{}
\SetKwRepeat{Do}{do}{while}

\SetKwFunction{GraverComplexity}{GraverComplexity}
\SetKwFunction{FindFeasibleSolution}{FindFeasibleSolution}
\SetKwFunction{GraverBasis}{GraverBasis}
\SetKwFunction{BuildGamma}{BuildGamma}
\SetKwFunction{gammaBestStep}{gammaBestStep}
\SetKwFunction{DynamicProgramStates}{DynamicProgramStates}
\SetKwFunction{ExhaustDirection}{ExhaustDirection}

$\vex_0 =$ \FindFeasibleSolution{$E,N,\veb,\vel,\veu$}, $i = 0$\;

\Do{$\vex_{i-1} \neq \vex_i$}{ \label{alg:outerlook}
  $\Gamma = \emptyset; \, j=0$,\, $i = i+1$\;
  \Do{$\veg_\lambda \neq \mathbf{0}$}{ \label{alg:innerloop}
    $\lambda = c^j$\;
    $\veg_\lambda = \min \left\{ \vew \veh \mid A \veh = \mathbf{0}, \, \vel \leq \vex + \lambda \veh \leq \veu,\, \|\veh\|_1 \leq \gc, \, \veh \in \Z^{Nt} \right\}$\;

    $\lambda' = \ExhaustDirection(\veg_\lambda)$\;
    $\Gamma = \Gamma \cup \{\lambda'\}, \, j = j+1$\;
  }
  $\vex_i = \vex_{i-1} + \argmin_{\left\{ \veg_\lambda \mid \lambda \in \Gamma \right\} } \vew \lambda \veg_\lambda$\;
}
\Return $\vex_{i}$

\caption{\label{alg:New}%
Pseudocode of our new heuristic algorithm.
The algorithm is exact if $\gc \geq g_1(A) = \max_{\veg \in \G(A)} \|\veg\|_1$.
Note the two nested loops: we shall refer to them as the \emph{inner loop} which computes a $c$-approximate Graver-best step, and the \emph{outer loop} which repeatedly adds the computed step to the current solution $\vex_{i-1}$.
}
\end{algorithm}

\subsection{Software and Hardware}
We have implemented our solver in the SageMath computer algebra system~\cite{SageMath}.
This was a convenient choice for several reasons.
The SageMath system offers an interactive notebook-style web-based interface, which allows rapid prototyping and debugging.
Data types for vectors and matrices, Graver basis algorithms~\cite{4ti2}, and a unified interface for MILP solvers are also readily available.
We have experimented with the open-source solvers GLPK~\cite{GLPK}, Coin-OR CBC~\cite{CBC}, and the commercial solver Gurobi~\cite{Gurobi} and have settled for using the latter, since it performs the best.
The downside of SageMath is that an implementation of the original dynamic program is likely much slower than a similar implementation in C; however this DP is impractical anyway as explained in Section~\ref{subsec:dp_ilp}.
Moreover, as we will evidence later, the overhead of SageMath in the construction of a MILP model is significant and for smaller instances (where~\eqref{AugIP} is not called many times) the time spent on constructing the MILP model dominates the runtime.
For random instance generation and subsequent data evaluation and graphing, we have used the Jupyter notebook environment~\cite{Jupyter} and Matplotlib and Seaborn libraries~\cite{Matplotlib,seaborn}.
The computations were performed on a computer with an Intel\textregistered ~Xeon\textregistered ~E5-2630 v3 (2.40GHz) CPU and 128 GB RAM.

\section{Testing Instances}
\subsection{Instances}
We choose two problems for which $N$-fold IP formulations were shown in the literature, namely the $Q || C_{\max}$ scheduling problem~\cite{KnopKoutecky:2017} and the \textsc{Closest String} problem~\cite{KnopKM:2017esa}.
Here we introduce both problems in their decision variants.

\prob{\textsc{Uniformly related machines makespan minimization ($Q || C_{\max}$)}}
{Set of $m$ machines $M$, each with a speed $s_i \in \N$.
A set of $n$ jobs $J$, each with a processing time $p_j \in \N$.
A target makespan $B$.}
{Is there an assignment of jobs $J$ to $m$ machines such that the time when the last job finishes (the makespan) is at most $B$? Here, a job $j$ scheduled on a machine~$i$ takes time $p_j / s_i$ to execute.}

\prob{\textsc{Closest String}}
{A set of $k$ strings $s_1, \dots, s_k$ of length $L$ over an alphabet $\Sigma$ and a positive integer $d$.}
{Is there a string $y \in \Sigma^L$ such that $\max_{i=1}^k d_H(s_i, y) \le d$, where $d_H$ is the Hamming distance?}

In the rest of this section we present $N$-fold IP models we used in our study and the describe how we generate random instances.

\subsection{Scheduling}
We observe that $Q||C_{\max}$ is equivalent to the multi-sized bin packing problem, where we have $m$ bins of various capacities instead of $m$ machines of different speeds, and we adopt this view as it is more convenient.
We also view it as a \emph{high-multiplicity} problem where the items are not given explicitly as a list of item sizes, but succinctly by a vector of item multiplicities.
Because Algorithm~\ref{alg:New} is primarily an optimization algorithm, we follow the standard approach~\cite[Lemma 3.8]{HemmeckeOR13} and turn the feasibility problem into an auxiliary optimization instance in which finding a starting feasible solution is easy.
However, the naive approach~\cite[Lemma 3.8]{HemmeckeOR13} would almost double the dimension, which is not necessary in the specific case of $Q || C_{\max}$.
Instead, we introduce an auxiliary machine onto which all jobs are initially scheduled, and the objective is to minimize the number of jobs scheduled on this machine.
If a solution is found with no jobs scheduled on this auxiliary machine, it corresponds to an admissible schedule with makespan at most $B$.

\subparagraph*{$N$-fold IP Model}
Let $\vep = (p_1, \dots, p_k)$ be the vector of item sizes, let $\ven = (n_1, \ldots, n_k)$ be the vector of item multiplicities, $n = \sum_{j = 1}^k n_j$, and let $s_1, \ldots, s_m$ be speeds of the machines in the instance of $Q || C_{\max}$.
We use the following ILP model for $Q || C_{\max}$ with fixed makespan $B$.
We have $km$ integral variables $x^i_j$ with $i \in [m]$ and $j \in [k]$ to express the number of jobs of type $j$ scheduled on machine~$i$.
Furthermore, we introduce a variable $x^0_j$ expressing the number of unscheduled jobs of type $j$ for $j \in [k]$.
As already pointed out we minimize the number of unscheduled jobs.
\begin{align*}
   \text{minimize} & \ \sum_{j = 1}^k x^0_j & \\
 \text{subject to}
& \ \sum_{i=0}^m x_j^i = n_j \qquad    & \forall 1 \leq j \leq k      \\
 & \ \sum_{j = 1}^k p_j x_j^i \le s_i \cdot B      \qquad      & \forall 1 \leq i \leq m      \\
 & \ \sum_{j = 1}^k p_j x_j^0 \le n \cdot p_k                                       \\
 \text{where}
& \ 0 \le x^i_j \le n_j \qquad & \forall i = 0,\ldots,m \forall j = 1, \ldots, k
\end{align*}
Here, we have essentially added a ``penalty machine'' which runs fast enough so that it is possible to schedule all of the given jobs to this extra machine.
Now, it is straightforward to verify that this is indeed an $N$-fold IP model with $N = m+1$ in which the matrix $E_1$ is the identity matrix of size $k \times k$ and $E_2 = \vep$.

The input parameters of the instance generation are number of bins (or machines) $m$, the smallest and the largest capacities $S$ and $L$, respectively, item sizes $p_1, \dots, p_k$ and probability weights $w_1, \dots, w_k$, and a slack ratio $\sigma \in \R$ with $0\le \sigma \le 1$.
Let $W = \sum_{i=1}^k w_i$.
The instance is then generated as follows.
First, we choose $m$ capacities from $[S,L]$ uniformly at random.
This determines the total available time of the machines $C$.
The next goal is to generate items whose total size is roughly $\sigma \cdot C$.
We do this by repeatedly picking an item length from $p_1, \dots, p_k$, where $p_j$ is selected with probability $w_j / W$, until the total size of items picked so far exceeds $\sigma \cdot C$, when we terminate and return the generated instance.

\subsubsection*{Batch generation.}
We generate a batch of experimental instances from a list of parameters, which correspond to command line arguments of the batch generator.
The generated batch is a cartesian product of all possible choices of the parameters.
\begin{description}
\item[\texttt{machines}] A list\footnote{List refers to the list datatype of the Python programming language.} of integers, by default \texttt{[10,20,30,40,50,60,70,80,90,100]}, corresponding to choices of the number of machines (bins) $m$.
\item[\texttt{number\_job\_types}] A list of integers, by default \texttt{[4]}, corresponding to different choices of the number of types $k$.
\item[\texttt{slacks}] A list of floats, by default \texttt{[0.6,0.7,0.8]}, corresponding to choices of the slack ratio~$\sigma$.
\item[\texttt{p\_s}] A list of integers, by default \texttt{[5,6,7,8,9,10,11,12,13]}. For each number $\ell \in \texttt{p\_s}$, we compute the first $\ell$ primes and randomly pick a subset of size $k$ of them as the processing times $p_1 \leq \dots \leq p_k$. We set the weights $w_1 \geq \dots \geq w_k$ to be $p_k, \dots, p_1$, i.e., jobs of larger length occur with smaller probability.
Note that $\max\texttt{number\_job\_types} \le \min\texttt{p\_s}$ must hold.
(We pick processing times which are primes because this easily guarantees that the set of $p_i$'s is coprime and thus the instance cannot be trivially reduced to an instance with smaller $p_{\max}$.)
\item[\texttt{count\_for\_each\_p}] An integer, by default \texttt{3}. For each choice of $\ell \in \texttt{p\_s}$ we make \texttt{count\_for\_each\_p} independent choices of the size $k$ subset of the first $\ell$ primes.
\end{description}

\subsection{Closest String}
The random instance is generated exactly as done by Chimani et al.~\cite{ClosestString}: first, we generate a random ``target'' string $y \in \Sigma^L$ and create $k$ copies $s_1, \dots, s_k$ of it; then, we make $\alpha$ random changes in $s_1, \dots, s_k$.
This way, we have an upper bound $\alpha$ on the optimum.
The input parameters of the instance generation are thus $k, L, \Sigma$, the distance ratio $r$ such that $\alpha = \nicefrac{n}{r}$, and a distance factor $\delta$, $0 \leq \delta \leq 1$, such that we ask whether there exists a string in distance $d=\delta \cdot \nicefrac{n}{r}$.
Thus for $\delta=1$ we are guaranteed that the answer is \textsc{Yes} while for $\delta=0$ the answer is almost surely \textsc{No}.
Again, we solve an auxiliary optimization instance where we essentially start with a string of ``all blanks'', where we set the Hamming distance between the blank and any character in $\Sigma$ to $0$.
Then, we try to fill in all the blanks while staying in the specified distance $d$; the objective is thus the remaining number of blanks.

\subparagraph*{$N$-fold IP Model}\footnote{The model is taken from~\cite{KnopKM:2017esa}.}
Let $\left(s_1, \ldots, s_k, d \right)$ be an instance of the \textsc{Closest String} problem, where all of the strings $s_1, \ldots, s_k$ are of length $n$ and taken from alphabet $\Sigma$.
We assume the given instance is already preprocessed, that is, $|\Sigma| \le k+1$ (the plus one comes from the presence of the blank symbol).
We call $k$-tuples of symbols in $\Sigma$ a \emph{configuration} and denote the set of all configurations $\mathcal{C}$.
An input position $i \in [n]$ has a configuration $C \in \mathcal{C}$ if $s_j[i] = C[i]$ for all $j = 1, \ldots, k$.
For a configuration $C \in \mathcal{C}$ by $n_C$ we denote the number of input positions having configuration $C$.
Notice now that our task is to decide for each configuration $C \in \mathcal{C}$ how many times we are going to use a character $\sigma \in \Sigma$ in the output sting $y$.
To that end we introduce integral variables $x_{C,\sigma}$ for each configuration $C \in \mathcal{C}$ and each character $\sigma \in \Sigma$.
Then, we introduce some auxiliary variables (all of them will be set to $0$ using the box constraints) in order to maintain the $N$-fold format and design a valid model with $N = \left|\mathcal{C}\right| \le k^k$.
To see this, notice that we have to compute the distance of $y$ to every string $s_i$ in the input.
Let $C\in\mathcal{C}$ be a configuration and let $D_C \in \{0,1\}^{k \times |\Sigma|}$ be the matrix whose columns we index by elements of $\Sigma$ with $D_C(i,\sigma) = d_H(C[i],\sigma)$, that is, the matrix $D_C$ describes the Hamming distance of the configuration $C$ if we decide to assign $\sigma$ once in the output string $y$.
We stress here that, since $\Sigma$ contains the blank symbol, $D_C$ contains the all zero column in the corresponding position corresponding.
Finally, we let $D = \left( D_{C_1} \mid \cdots \mid D_{C_{|\mathcal{C}|}} \right)$ be a matrix in which we collect all of the above defined distance matrices.
Let $t$ be the number of columns of the matrix $D$.
For each configuration $C \in \mathcal{C}$ we introduce a vector of variables $\vex^C$ of length $t$ whose entries we index $x_{\bar{C},\sigma}$; we set the box constrains to
\[
0 \le x_{\bar{C},\sigma} \qquad \forall \bar{C} \in \mathcal{C}, \forall \sigma \in \Sigma
\qquad\textrm{and}\qquad
x_{\bar{C},\sigma} \le 0 \qquad \forall \bar{C} \in \mathcal{C}\setminus\{C\}, \forall \sigma \in \Sigma \,.
\]
Now, the global conditions are
\[
  \sum_{C \in \mathcal{C}} D \vex^{C} \le \ved \,,
\]
where $\ved = (d, \ldots, d)$ is a vector of length $k$.
Finally, we set the local conditions
\[
  \sum_{\bar{C} \in \mathcal{C}} \sum_{\sigma \in \Sigma} x_{\bar{C},\sigma} = n_C \qquad\qquad \forall C \in \mathcal{C}
\]
and the objective function
\[
  \min \sum_{C \in \mathcal{C}} \sum_{\bar{C} \in \mathcal{C}} x^C_{\bar{C}, \lambda} \,,
\]
where $\lambda$ is the blank symbol.
This finishes the description of the used $N$-fold IP model.

\subsubsection*{Batch generation.}
The list of parameters for batch generation is the following:
\begin{description}
\item[\texttt{str\_len}] A list of integers, by default \texttt{[500,1000,2000,4000,8000,16000]}, corresponding to choices of $L$.
\item[\texttt{str\_num}] A list of integers, by default \texttt{[3,4,5,6]}, corresponding to choices of $k$.
\item[\texttt{ratio}] A list of integers, by default \texttt{[2,3,4,7,10,15]}, corresponding to choices of $r$.
\item[\texttt{sigma}] A list of integers, by default \texttt{[2,3,4,5]}, corresponding to choices of $|\Sigma|$.
\item[\texttt{distance\_factor}:] A list of floats, by default \texttt{[0.1,0.15,0.2,0.25,0.3,0.5,0.7]}, corresponding to choices of $\delta$.
\end{description}
We generate an instance for each parameter tuple from the cartesian product of all the lists above.

\subsection{Common Parameters}
Here we describe parameters which are common to both instance types ($Q || C_{\max}$ and \textsc{Closest String}).
For each generated instance we run the iterative algorithm for various choices of the augmentation strategy $\Gamma \in \{\Gamma_{\text{any}}, \Gamma_{\text{best}}, \Gamma_{\text{2-apx}}, \Gamma_{\text{5-apx}}, \Gamma_{\text{10-apx}}\}$ and the tuning parameter $\gc$.
The main parameters are thus
\begin{description}
\item[\texttt{gc\_values}] A list of integers, by default \texttt{[4,8,12,20,30,40,50,75,100]}, corresponding to choices of $\gc$.
\item[\texttt{gammas}] A list of strings, by default \texttt{["log2"]}, with other options being \texttt{"unit"}, \texttt{"best"}, \texttt{"log5"}, and \texttt{"log10"}, corresponding to the choices of $\Gamma$.
\end{description}

The parameter \texttt{logdir} (by default \texttt{logs}) determines the target directory to store the logs.
The directory will have subdirectories according to the dimension $Nt$ on the first level, subdirectories according to different $\Delta$ (maximum coefficient) on the second level, and subdirectories for each problem instance on the third level.
Finally, each instance directory contains one \texttt{.log} and one \texttt{.pickle} (protocol version 2) file for each choice of $\gc$ and $\Gamma$.
The parameter \texttt{instance\_type} is one of \texttt{sched} (default) or \texttt{cs}, for $Q || C_{\max}$ or \textsc{Closest String}, respectively.
Parameters \texttt{augip\_timelimit} and \texttt{milp\_timelimit} are both integers determining the timelimit for the MILP solver, with the former one applying to the~\eqref{AugIP} instance and the latter one to when we call the solver on the original~\eqref{IP} instance.
Finally, passing \texttt{--disable\_nfold} turns off the iterative algorithm and only uses the MILP solver to solve the original~\eqref{IP} instance.

\section{Evaluation}
We first give an outline of the evaluation process, which is divided into three parts.

\subsubsection*{Qualitative Evaluation}
In the first part we begin with two main questions, specifically, how is the performance of the algorithm (both in terms of the number of iterations and the quality of the returned solution) influenced by:
\begin{enumerate}
\item the value of the tuning parameter $\gc$ and
\item the augmentation strategy $\Gamma$?
\end{enumerate}

Regarding our first question, theoretically we should see either an increase in the number of iterations, a decrease in the quality of the returned solution, or both.
However, the range of the tuning parameter $\gc$ is quite large: any number between $2$ and $g_1(A)$ is a valid choice, and in all our scenarios the true value of $g_1(A)$ exceeds $200$.
Thus, we are interested in the transition values of $\gc$ when the algorithm no longer finds the true optimum or when its convergence rate drops significantly.

Regarding our second question, there are two main candidates for the set of step-lengths $\Gamma$.
We can either use the ``best step'' construction $\Gamma_{\text{best}}$ of the original algorithm, which assures that we always make a Graver\hy{}best step before moving to the next iteration.
Or, we can use the ``approximate best step'' construction $\Gamma_{\text{2-apx}}$ of Theorem~\ref{thm:gammanfold}, which provides a $2$-approximate Graver\hy{}best step.
To make this comparison more interesting, we also consider $\Gamma_{5\text{-apx}}$ and $\Gamma_{10\text{-apx}}$ and also the trivial ``any step'' strategy where we always make the $1$-Graver\hy{}best step, which corresponds to taking $\Gamma_{\text{any}} = \{1\}$.
Recall that due to the trick of always exhausting the discovered direction, this strategy actually has a chance at quick convergence, unlike if we only made the step with $\lambda = 1$.

\subsubsection*{Quantitative Evaluation}
Later, we will quantify the relationship of several instance parameters such as the dimension, largest coefficient $\Delta$, number of columns of $E_1$, which is $t$, number of rows of $E_1$, which is $r$, number of bricks $N$, and tuning parameter $\gc$, to performance parameters such as optimality gap or convergence rate.
Recall that in both our scenarios we have $s=1$ and thus we do not mention this parameter further.

\subsubsection*{Towards Practical Applications}
Finally, we explore possible avenues to transfer our ideas to practice.
To that end, we ask ``on which instances could $n$-fold IP beat Gurobi?''
Due to the immense amount of attention dedicated to industrial MILP solvers we do not expect our ideas to lead to significant improvements across many kinds of instances; however, we do expect that there exist some special instances on which Gurobi performs poorly and could be outperformed by a newer implementation of our solver.

To this end, we study the relationship of several time measures (total time, time spent on augmentation calls, time taken by Gurobi to solve the instance etc.) to parameters such as dimension, $\Delta$, $r$, and $t$.

\subsection{Qualitative Evaluation}
Here we demonstrate the overall behavior of the algorithm on two selected instances (one for $Q||C_{\max}$ and one for \textsc{Closest String}); we encourage the reader to see the full data (incl. plots) at \url{https://github.com/katealtmanova/nfoldexperiment}.


We chose two instances among the tested ones as representatives of the overall behavior:
\begin{itemize}
  \item
  A $Q || C_{\max}$ instance with parameters $m=60$, $S=215$, $L=12124$, item sizes $(3,7,17,41,43)$ (note this implies nontrivial $\Delta$), weights $(43,41,17,7,3)$, and $\sigma=0.6$.
The theoretical upper bound on $g_1(A)$ is $(rs \Delta + 1)^{\Oh(rs)}$~\cite[Lemma 3]{EisenbrandHK:2018}, and here we have $r=5$, $s=1$ and $\Delta = 43$; thus, without computing $g_1(A)$ exactly, we should consider it to be at least $(5 \cdot 43 + 1)^{5} \approx 4.7 \cdot 10^{11}$.
  \item
  A \textsc{Closest String} instance with parameters $k=3$, $|\Sigma|=4$, $L=8000$, $r=4$ and $\delta = 0.3$.
The $N$-fold model has $r=3$, $s=1$ and $\|A\|_\infty = 1$, thus, without computing $g_1(A)$ exactly, we should consider it to be at least $(2 \cdot 3)^3 = 216$.
\end{itemize}

\subsubsection*{Plots}
We use two types of plots to visualize our data.
First and only for the scheduling instance, we hand-picked four ``interesting'' values of $\gc$, namely $\gc=25,50,150$ and $1000$, and we give a line plot for each such value of $\gc$ and each augmentation strategy $\Gamma$.
The $x$ axis of each line plot corresponds to inner iterations (computations of a $\lambda$\hy{}Graver\hy{}best step).
The $y$ axis corresponds to objective values.
Each line plot contains two lines: a thin blue line marking each individual value computed in the inner loop, and a thick orange line marking the progress of the outer loop, i.e., the minimum over all steps computed in the individual outer iterations.

\begin{figure}[!h]
\centering
	\includegraphics[width=0.49\textwidth]{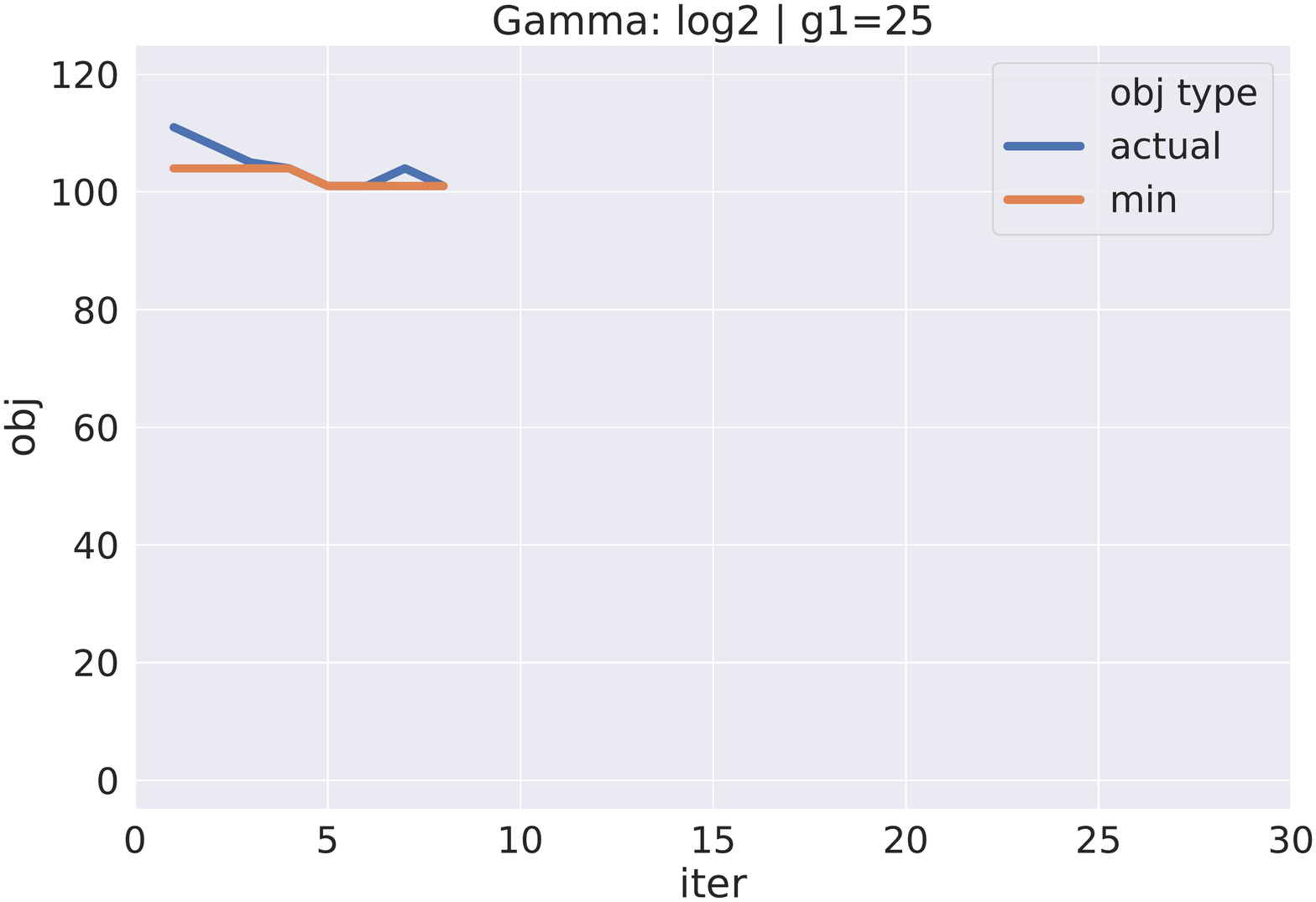}
	\includegraphics[width=0.49\textwidth]{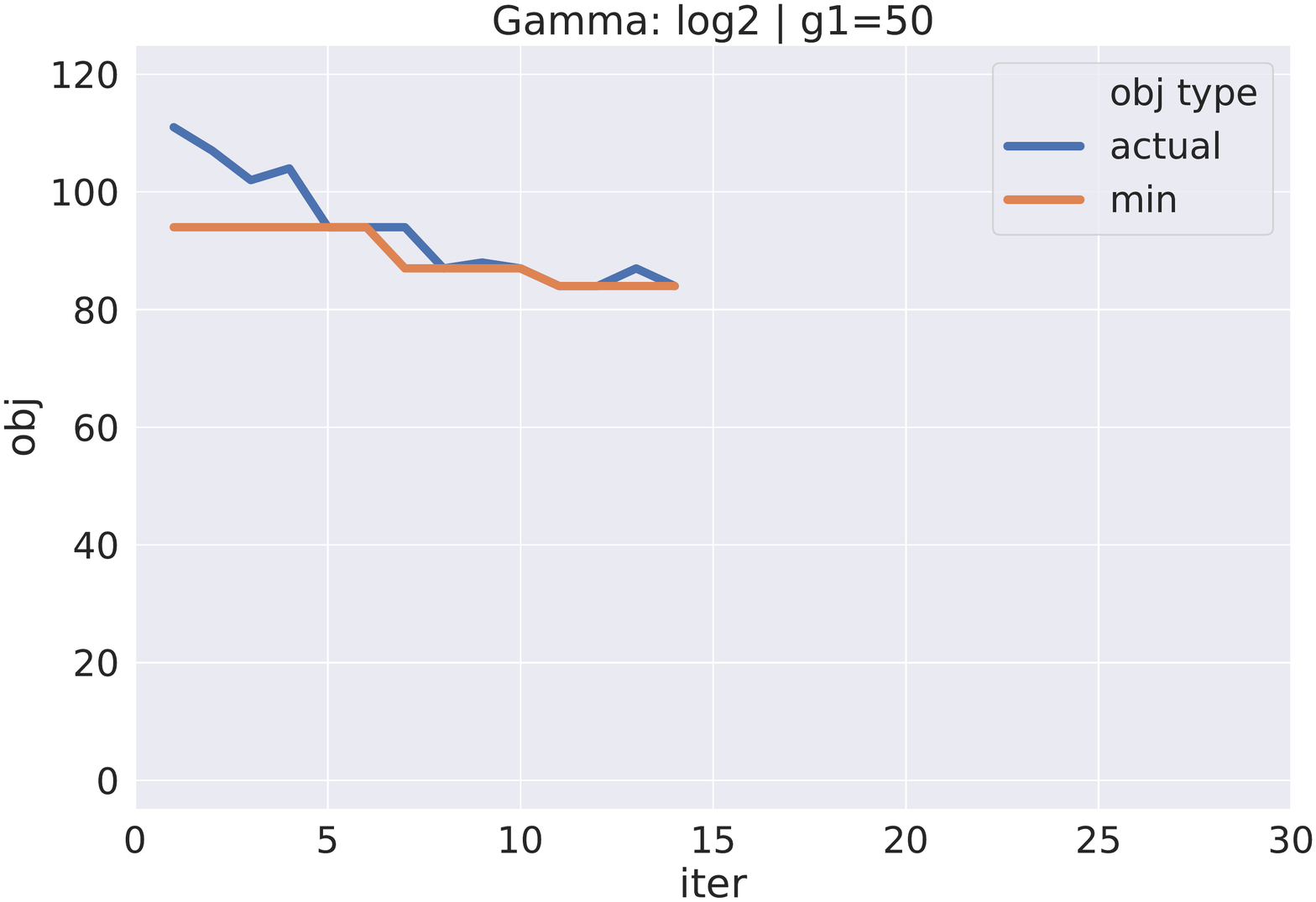}
	\includegraphics[width=0.49\textwidth]{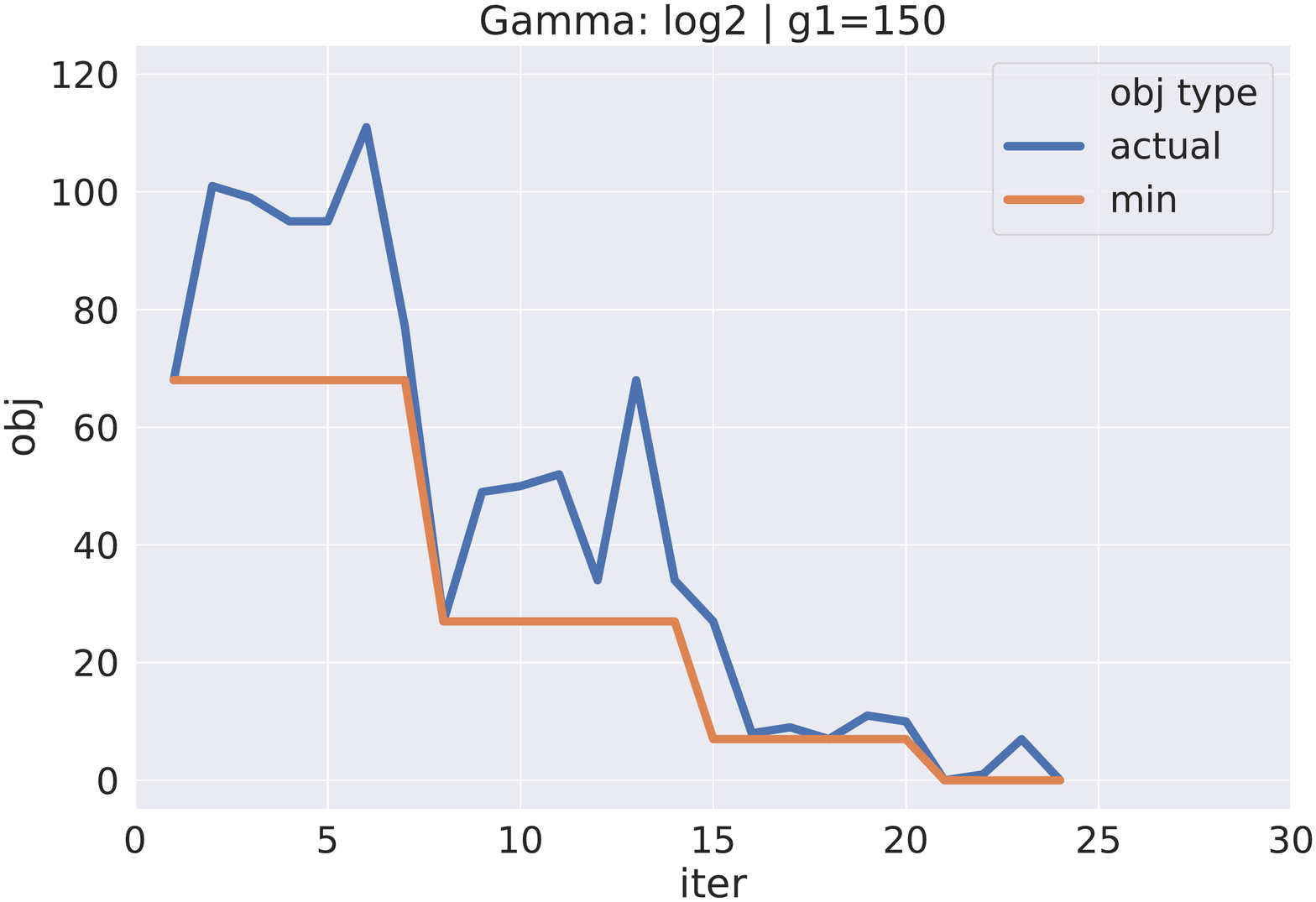}
	\includegraphics[width=0.49\textwidth]{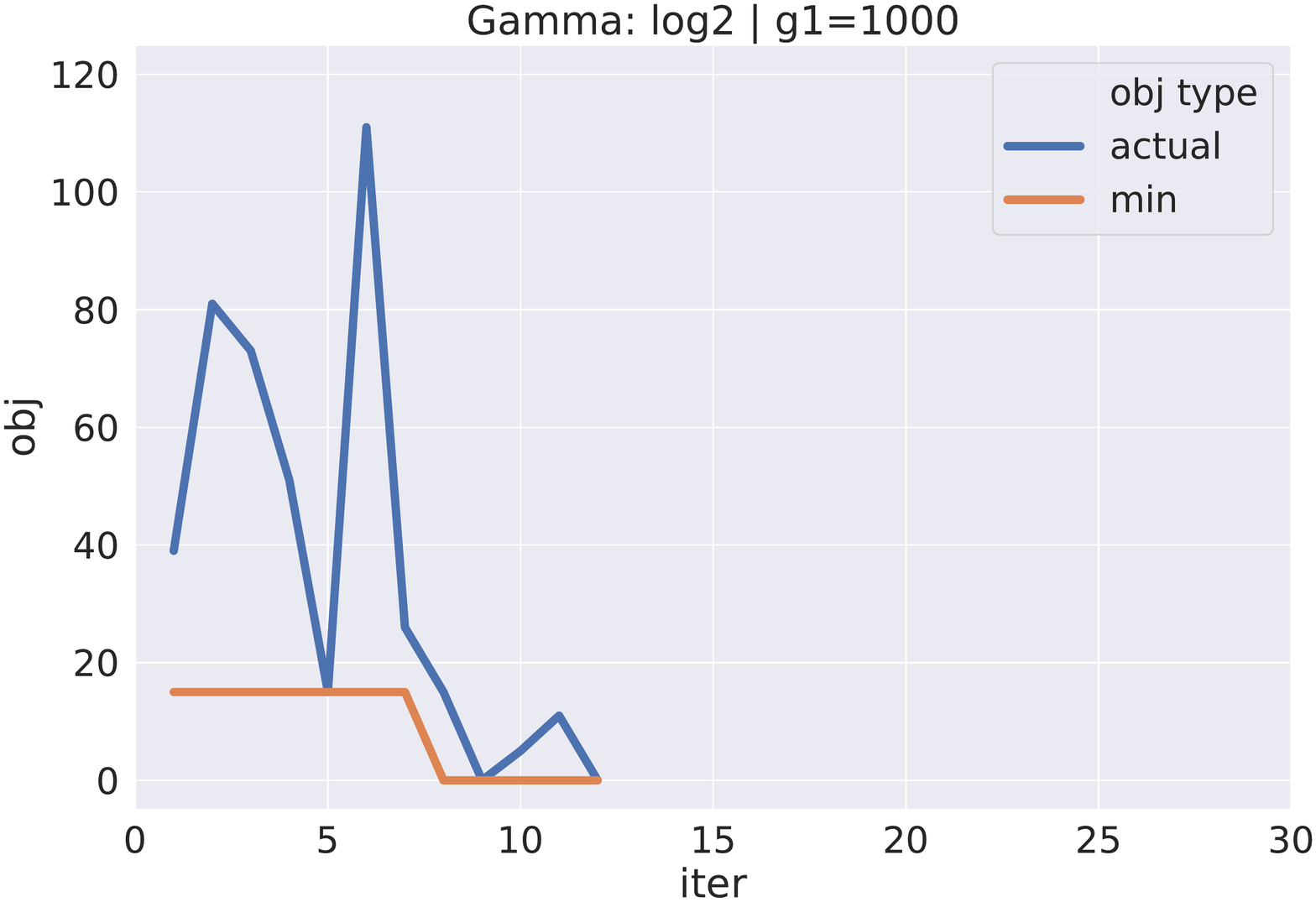}
	\caption{Augmentation strategy $\Gamma_{\text{2-apx}}$ on a $Q||C_{\max}$ instance. Blue line corresponds to inner loop values, orange line corresponds to steps actually made (outer loop). The number of iterations is measured in the inner loop (i.e., it is the number of~\eqref{AugIP} computations).\label{fig:sched_log2}}
\end{figure}

\begin{figure}[!h]
	\centering
	\includegraphics[width=0.49\textwidth]{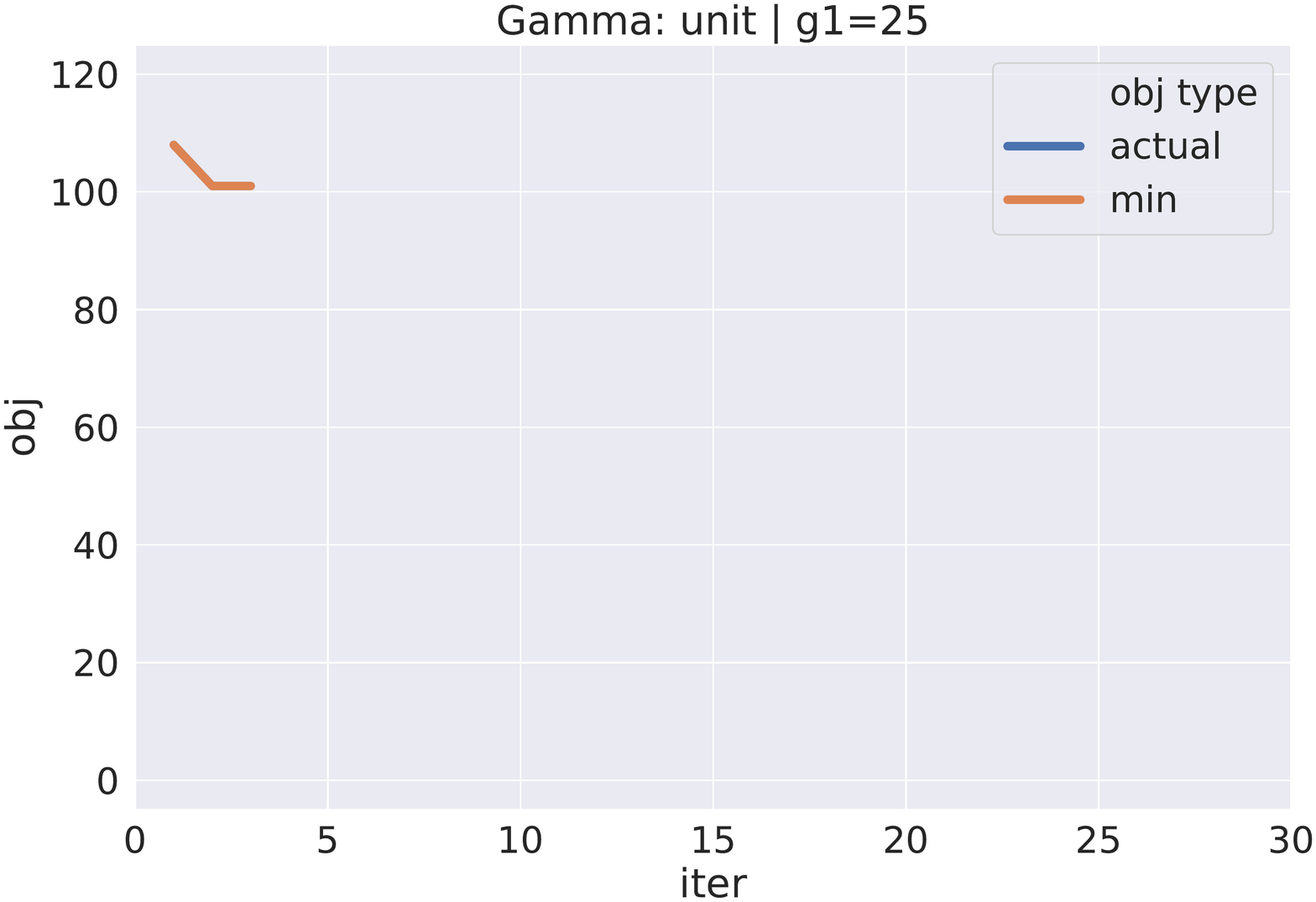}
	\includegraphics[width=0.49\textwidth]{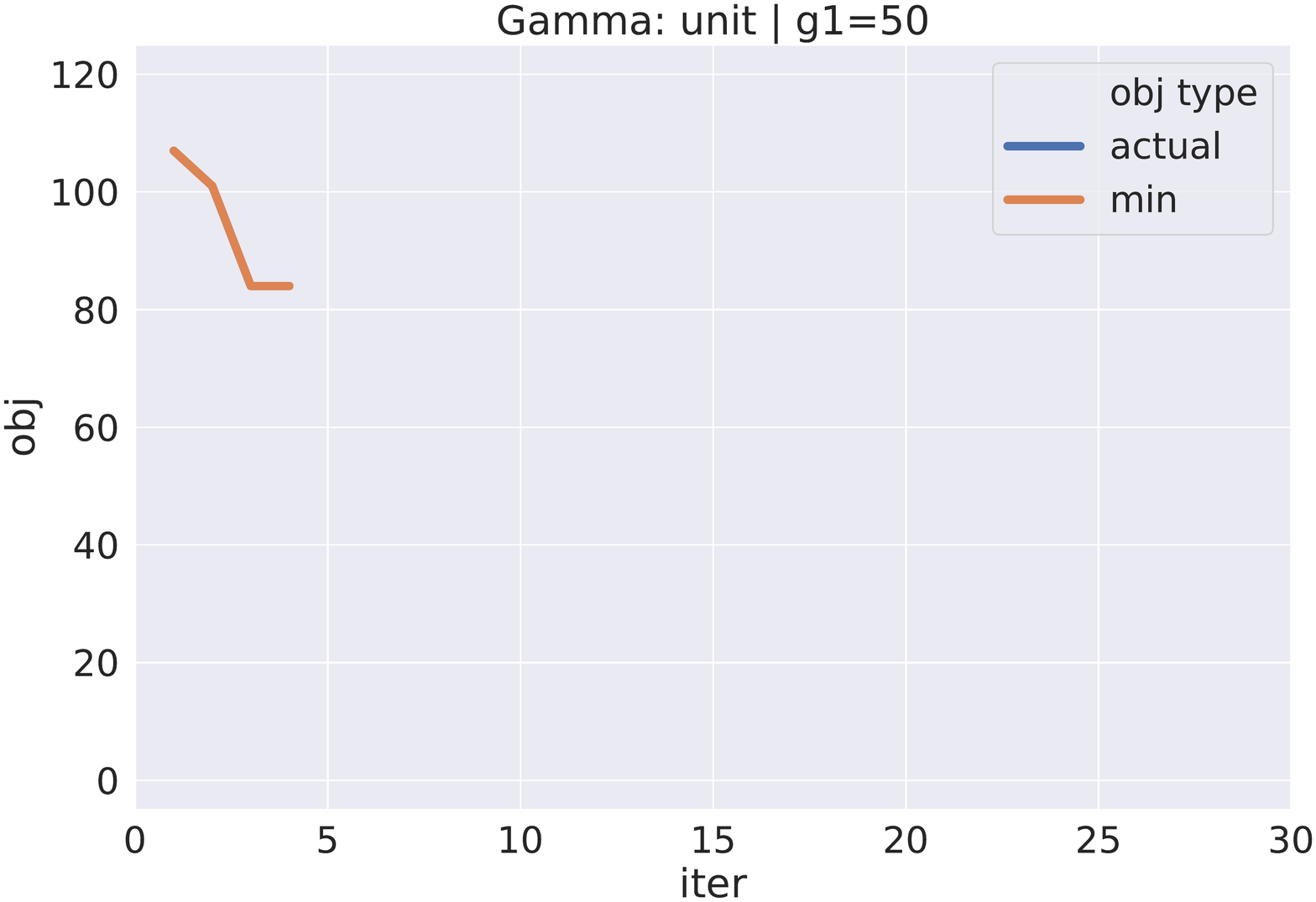}
	\includegraphics[width=0.49\textwidth]{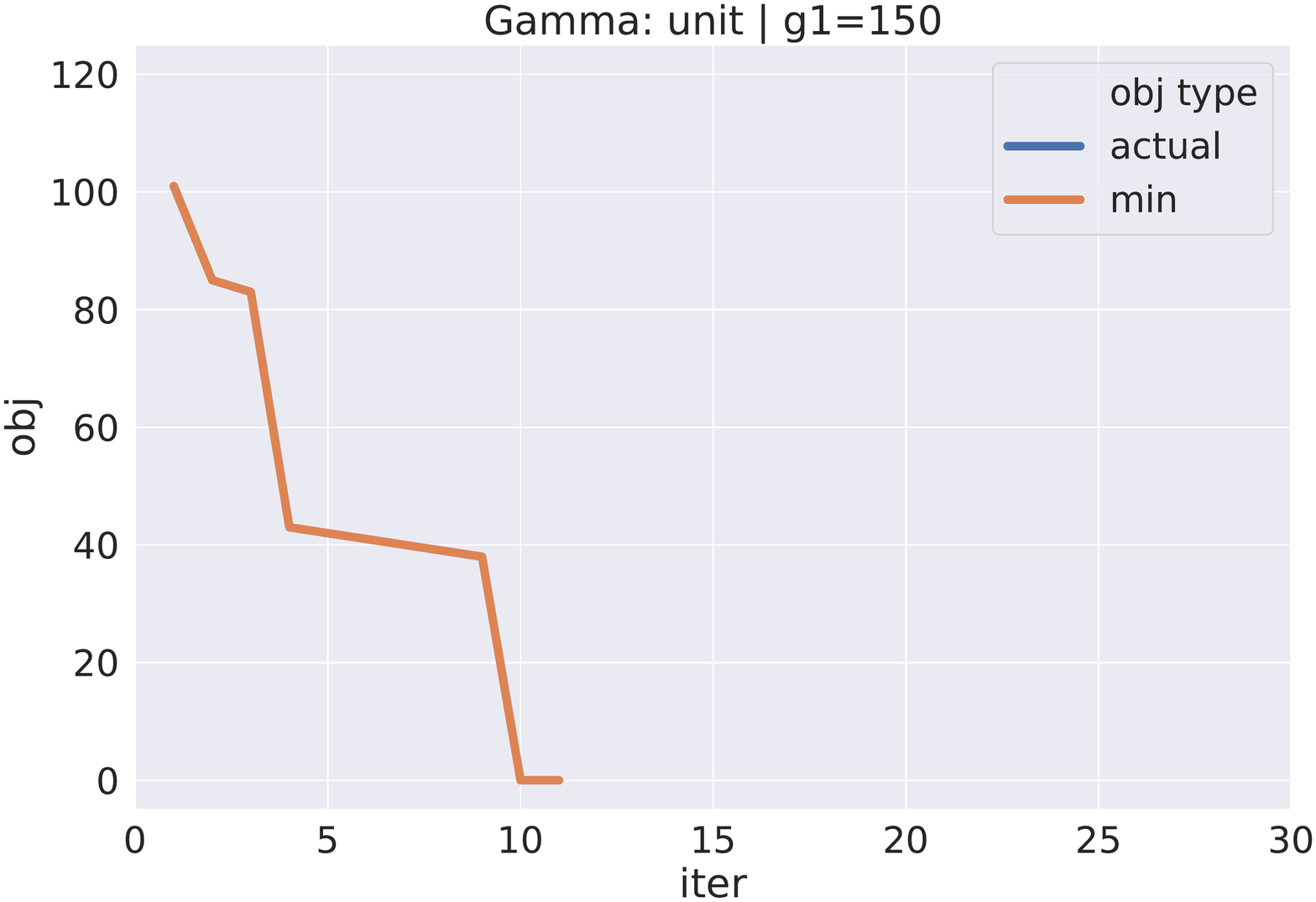}
	\includegraphics[width=0.49\textwidth]{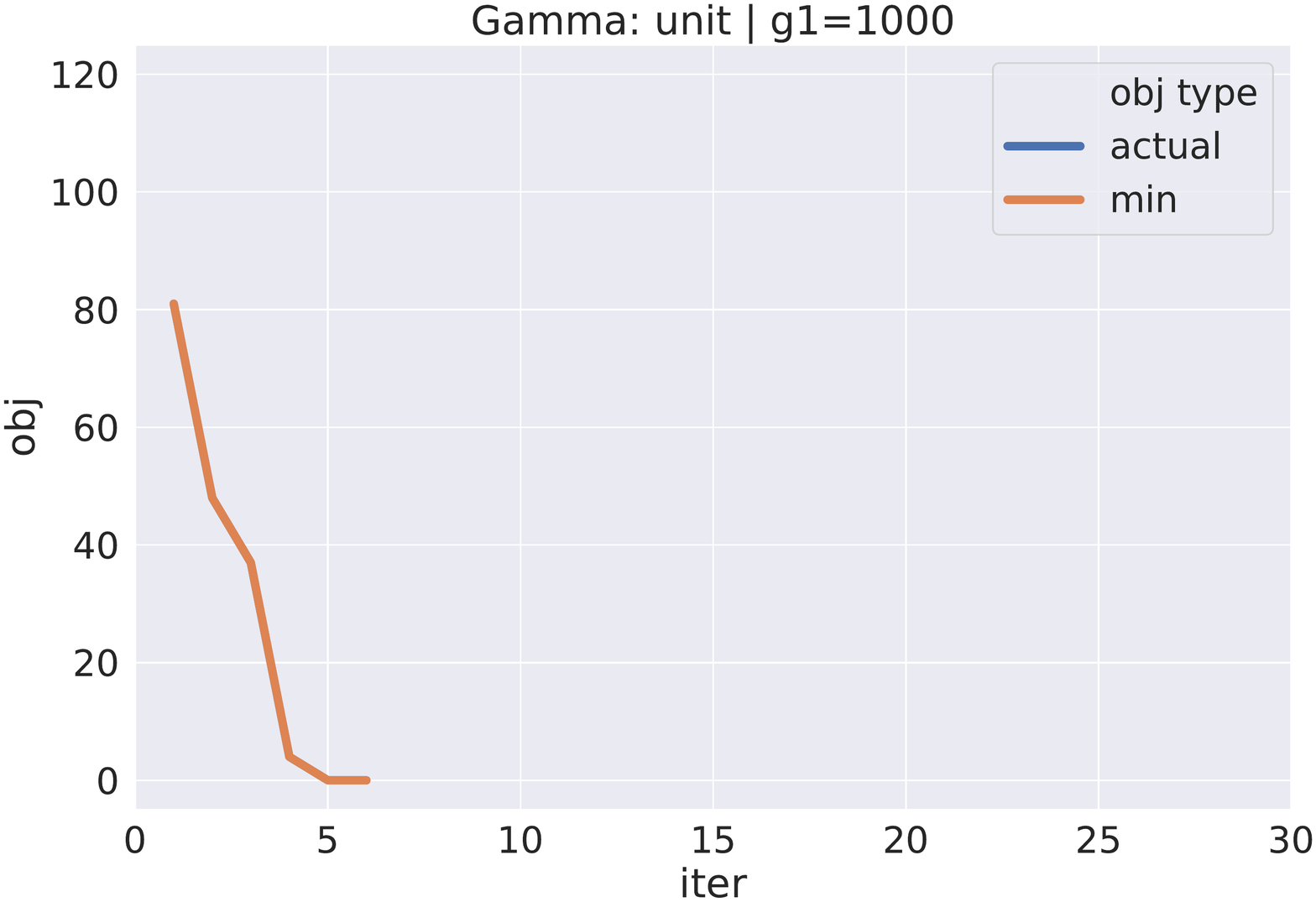}
	\caption{Augmentation strategy $\Gamma_{\text{any}}$ on a $Q||C_{\max}$ instance (for interpretation cf. Figure~\ref{fig:sched_log2}). \label{fig:sched_unit}}
\end{figure}

\begin{figure}[!h]
	\centering
	\includegraphics[width=0.49\textwidth]{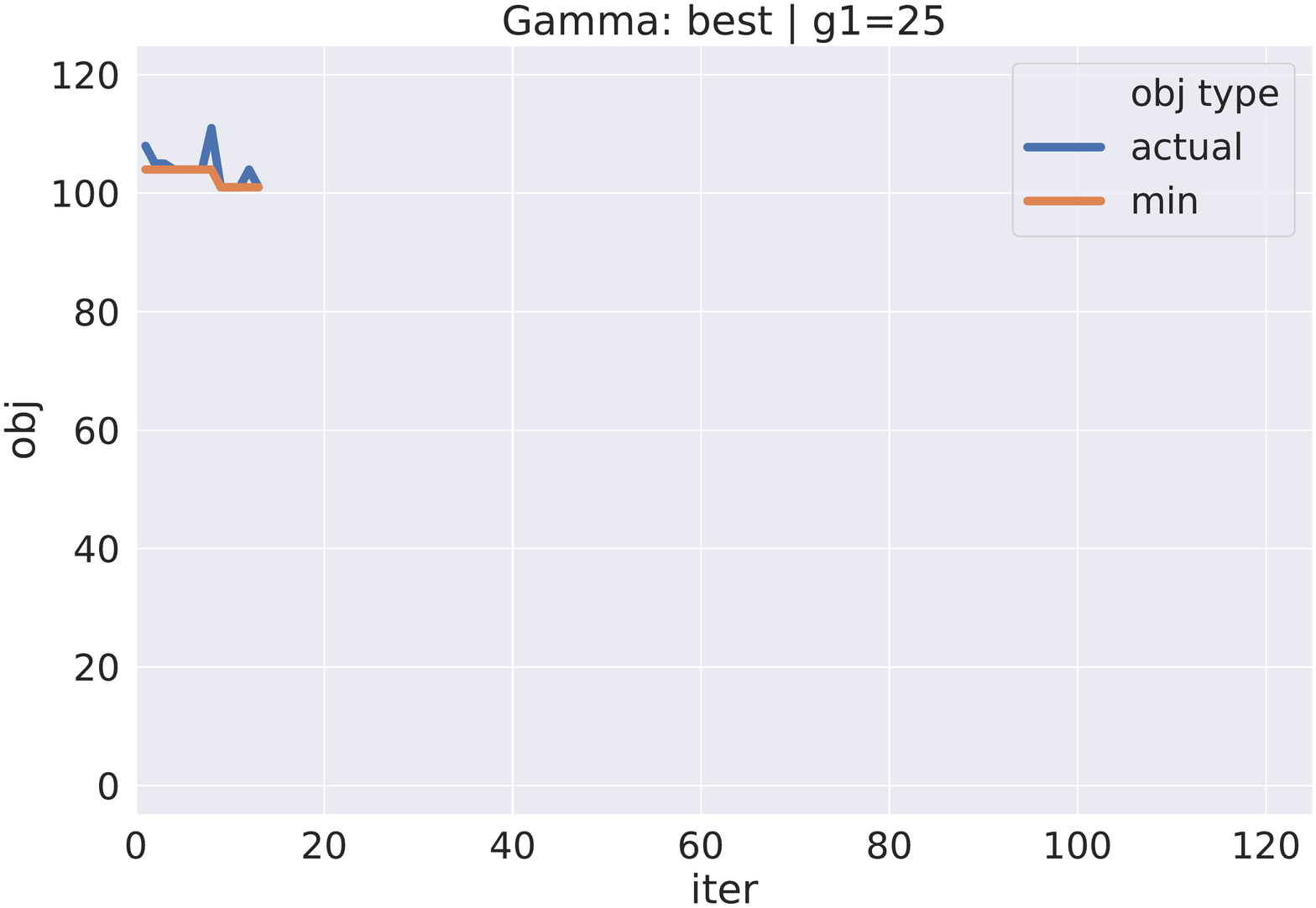}
	\includegraphics[width=0.49\textwidth]{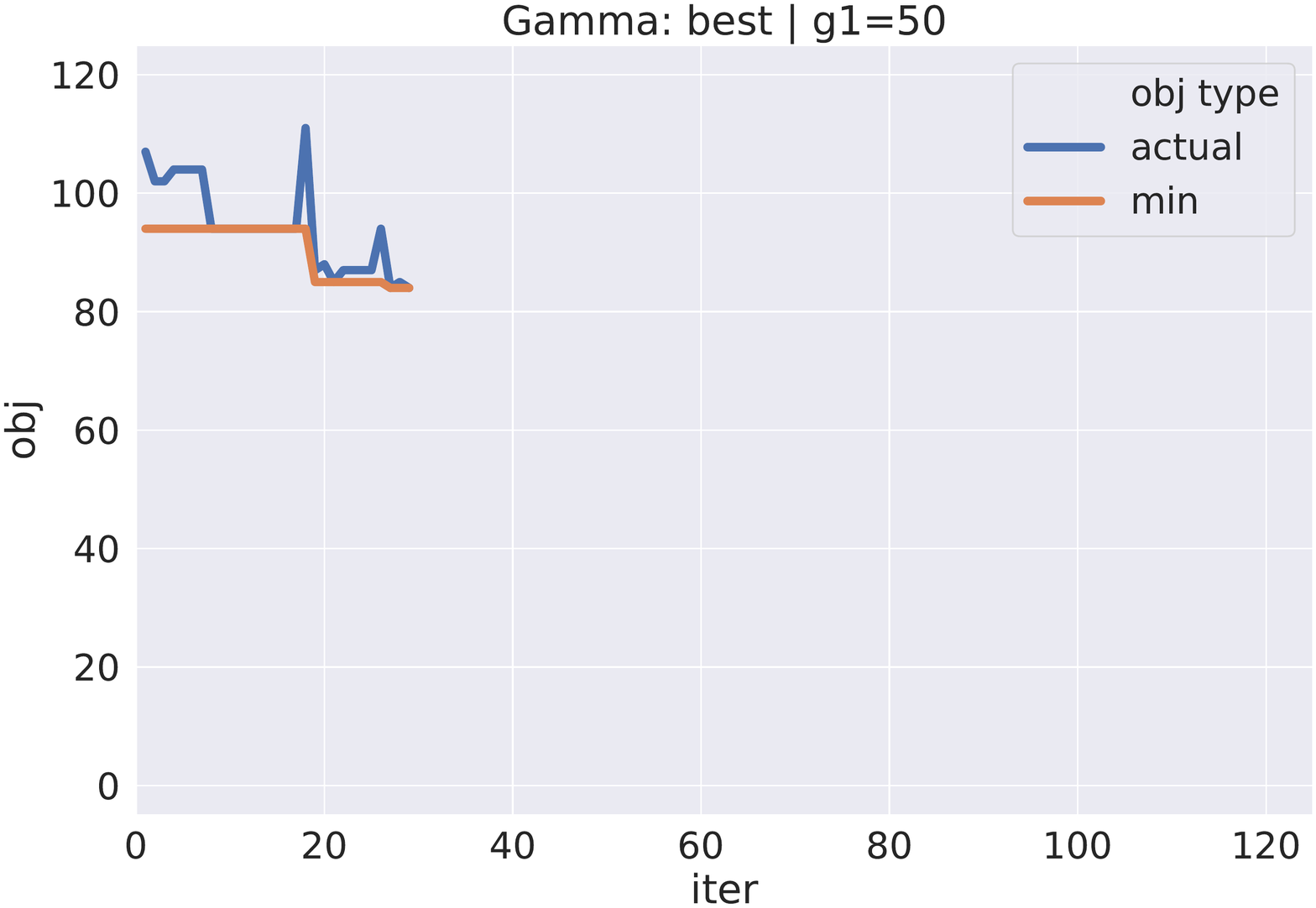}
	\includegraphics[width=0.49\textwidth]{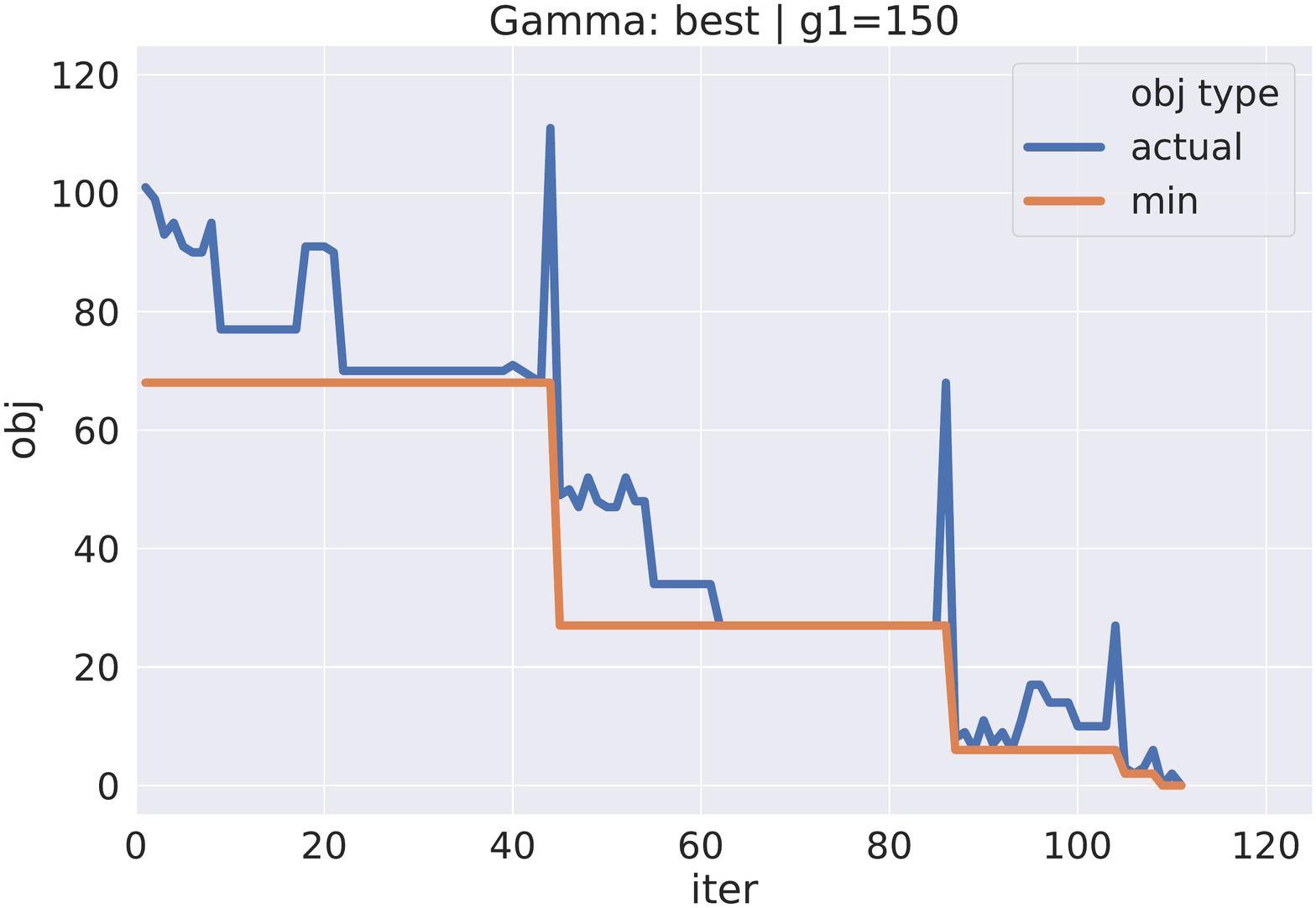}
	\includegraphics[width=0.49\textwidth]{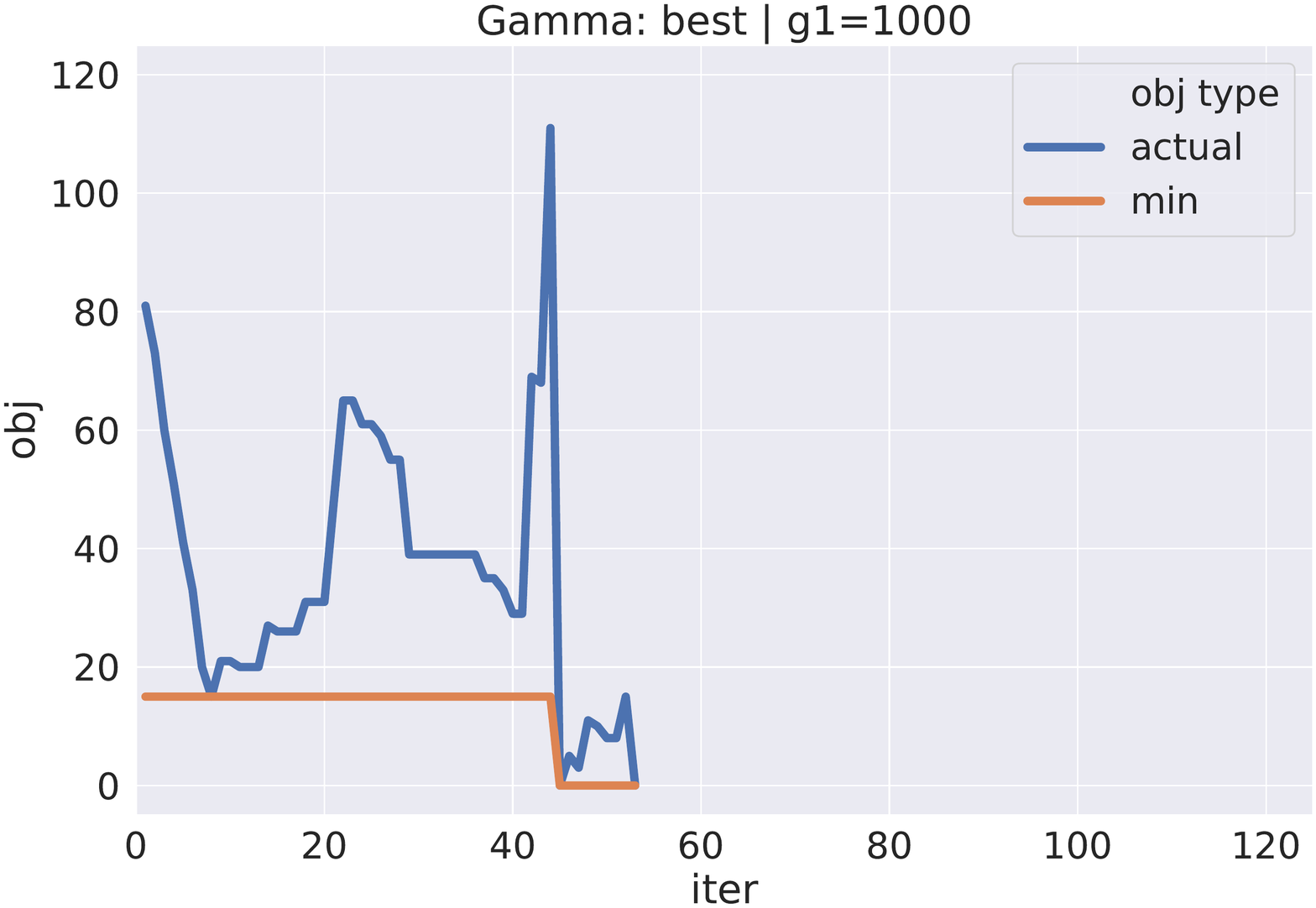}
	\caption{Augmentation strategy $\Gamma_{\text{best}}$ on a $Q||C_{\max}$ instance (for interpretation cf. Figure~\ref{fig:sched_log2}). \label{fig:sched_best}}
\end{figure}

%

The second type of plot (Figures~\ref{fig:sched_single} and~\ref{fig:cs_single}) is essentially obtained from the first type by considering all tested values of $\gc$ (not only the ``interesting'' values), discarding the thin (inner loop) lines, and stacking the remaining lines on top of each other, thus obtaining one line plot for each augmentation strategy $\Gamma$.

\begin{figure}[!h]
\centering
     \includegraphics[width=0.32\textwidth]{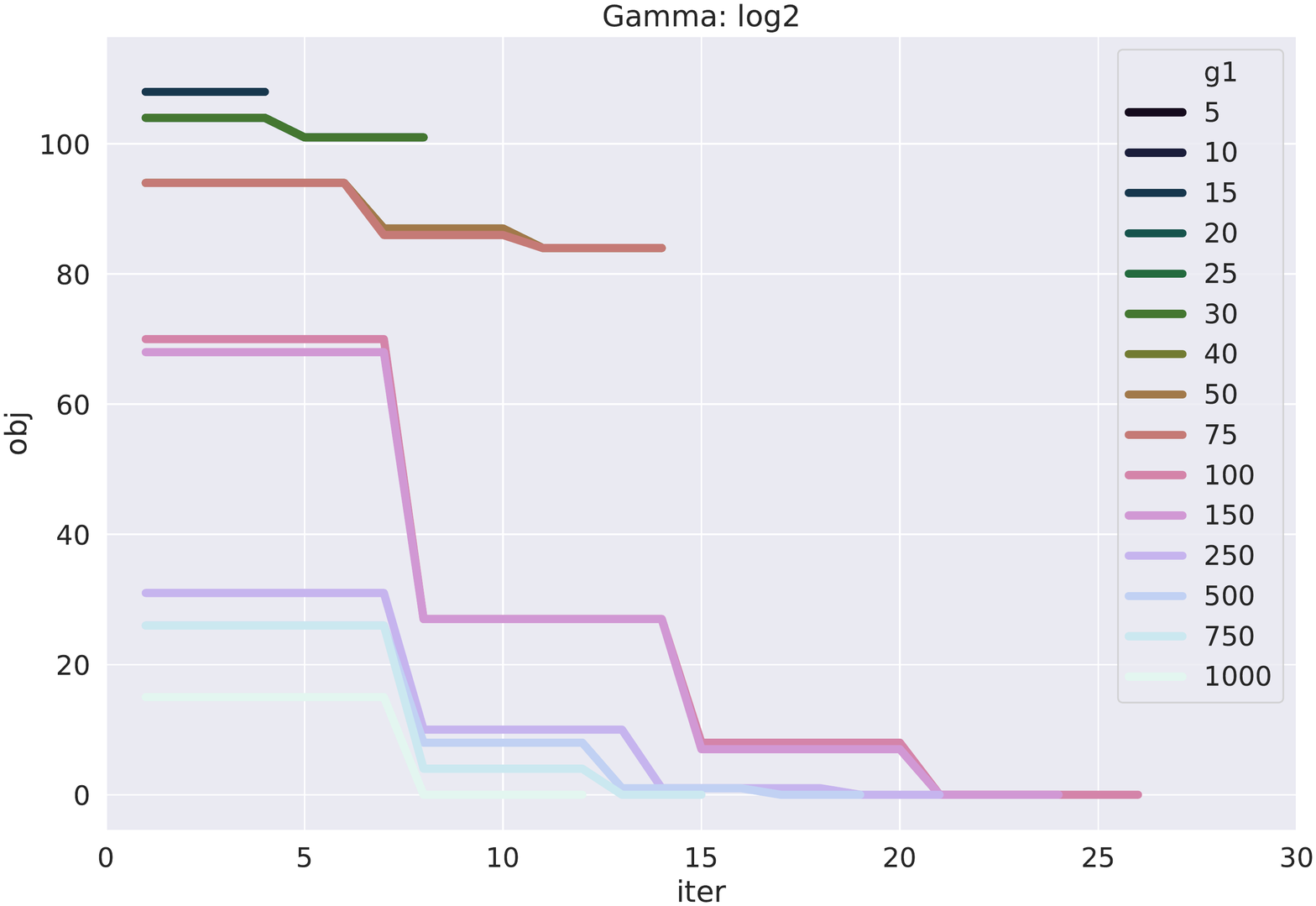}
     \includegraphics[width=0.32\textwidth]{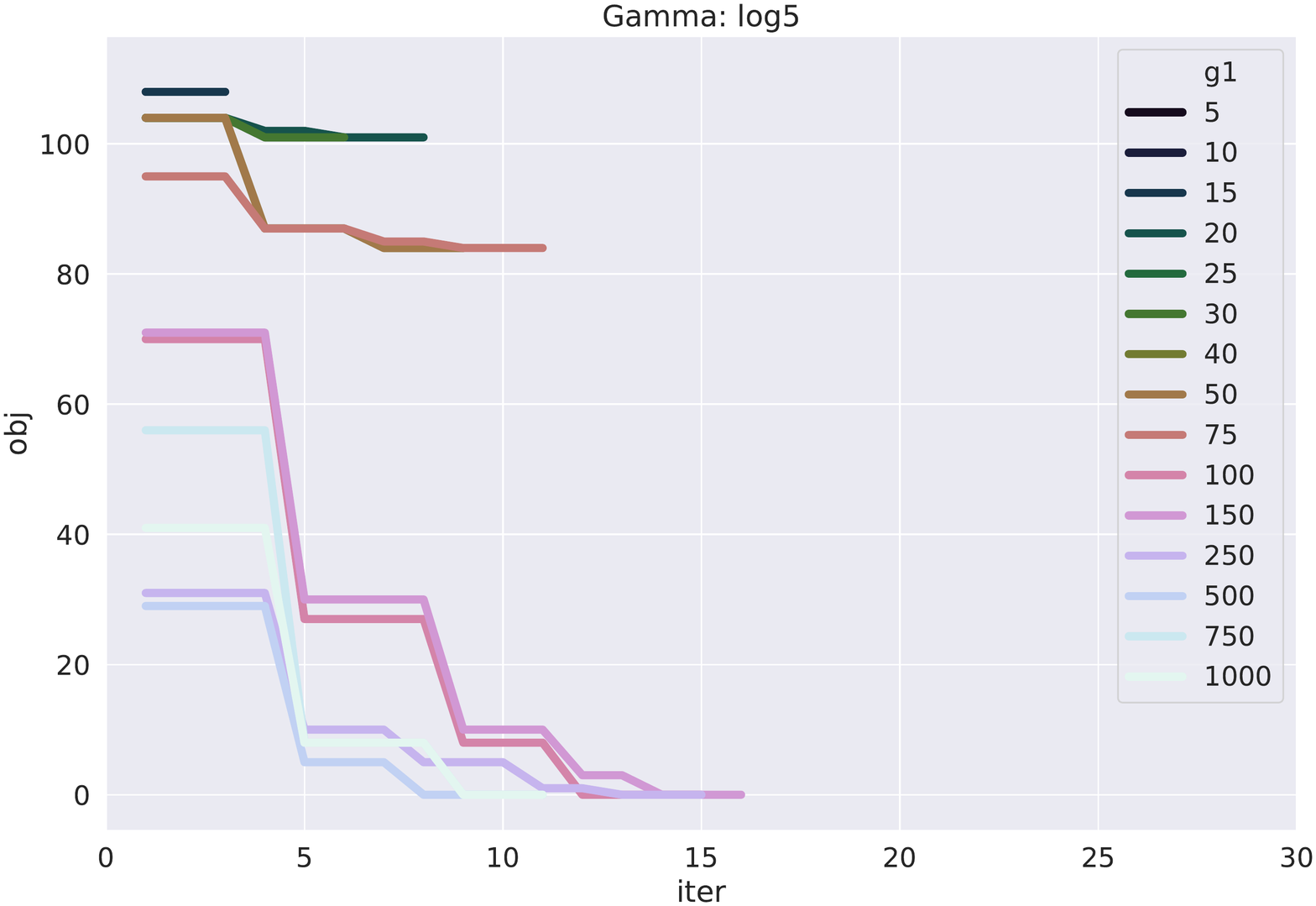}
     \includegraphics[width=0.33\textwidth]{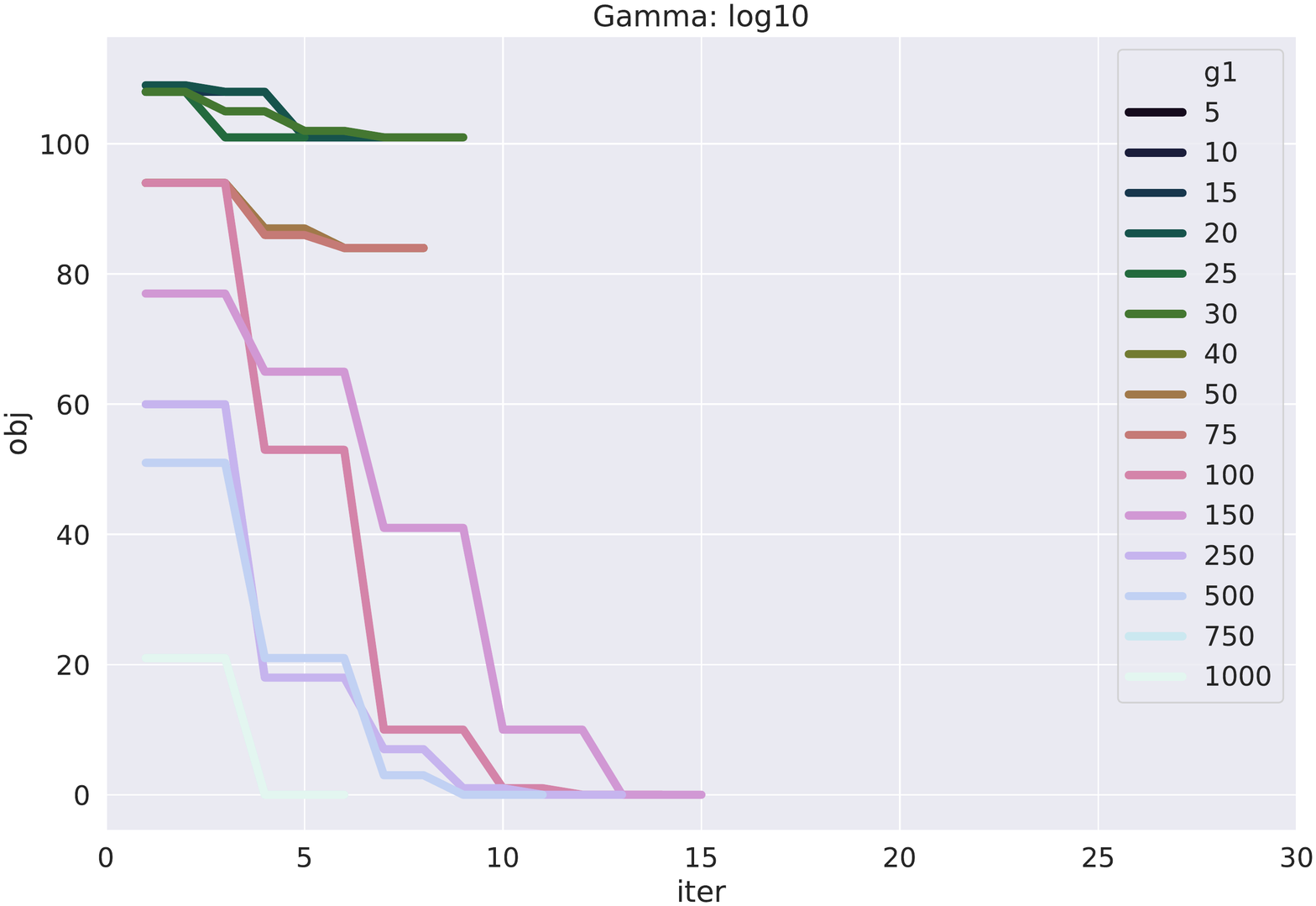}
     \includegraphics[width=0.49\textwidth]{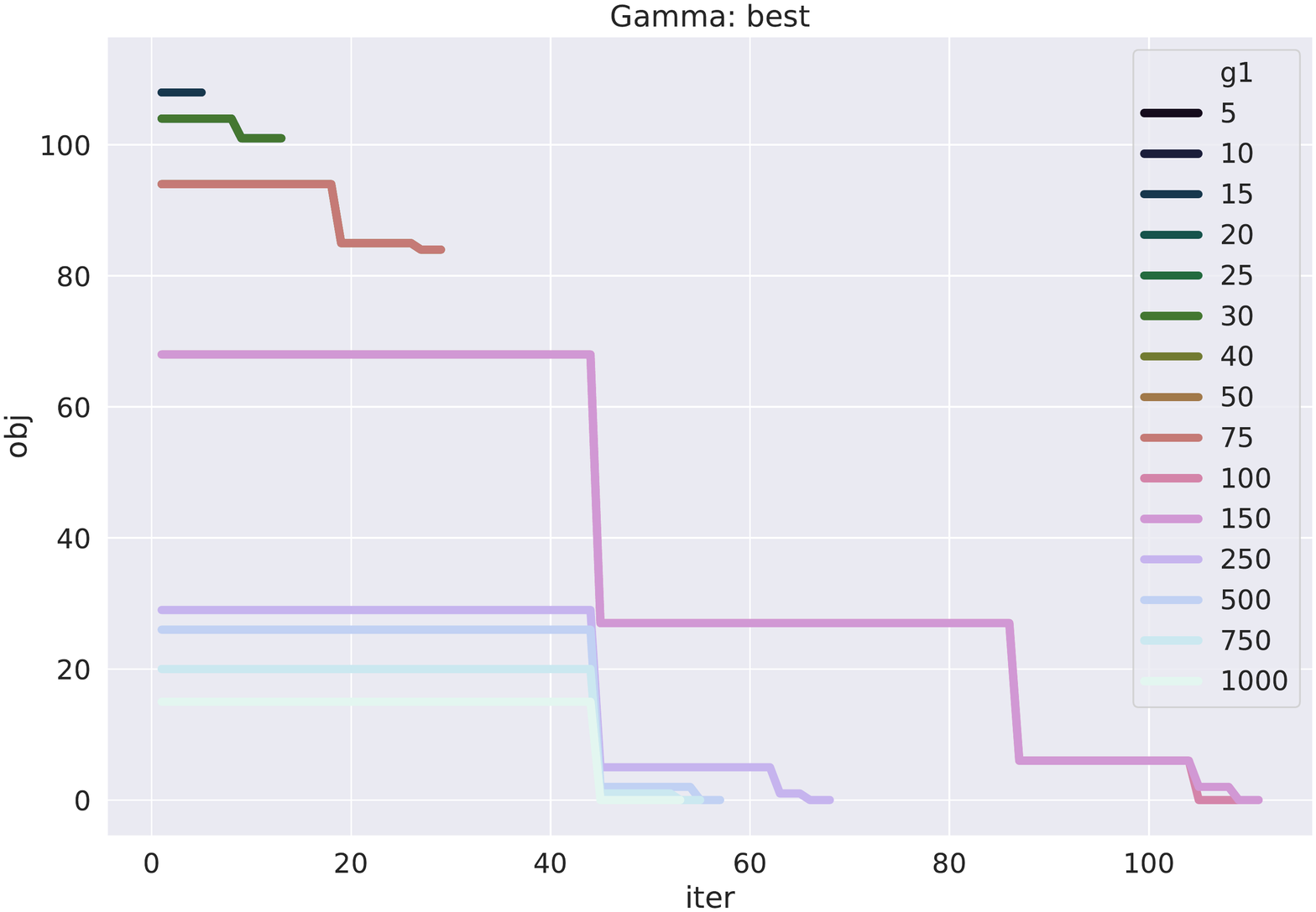}
     \includegraphics[width=0.49\textwidth]{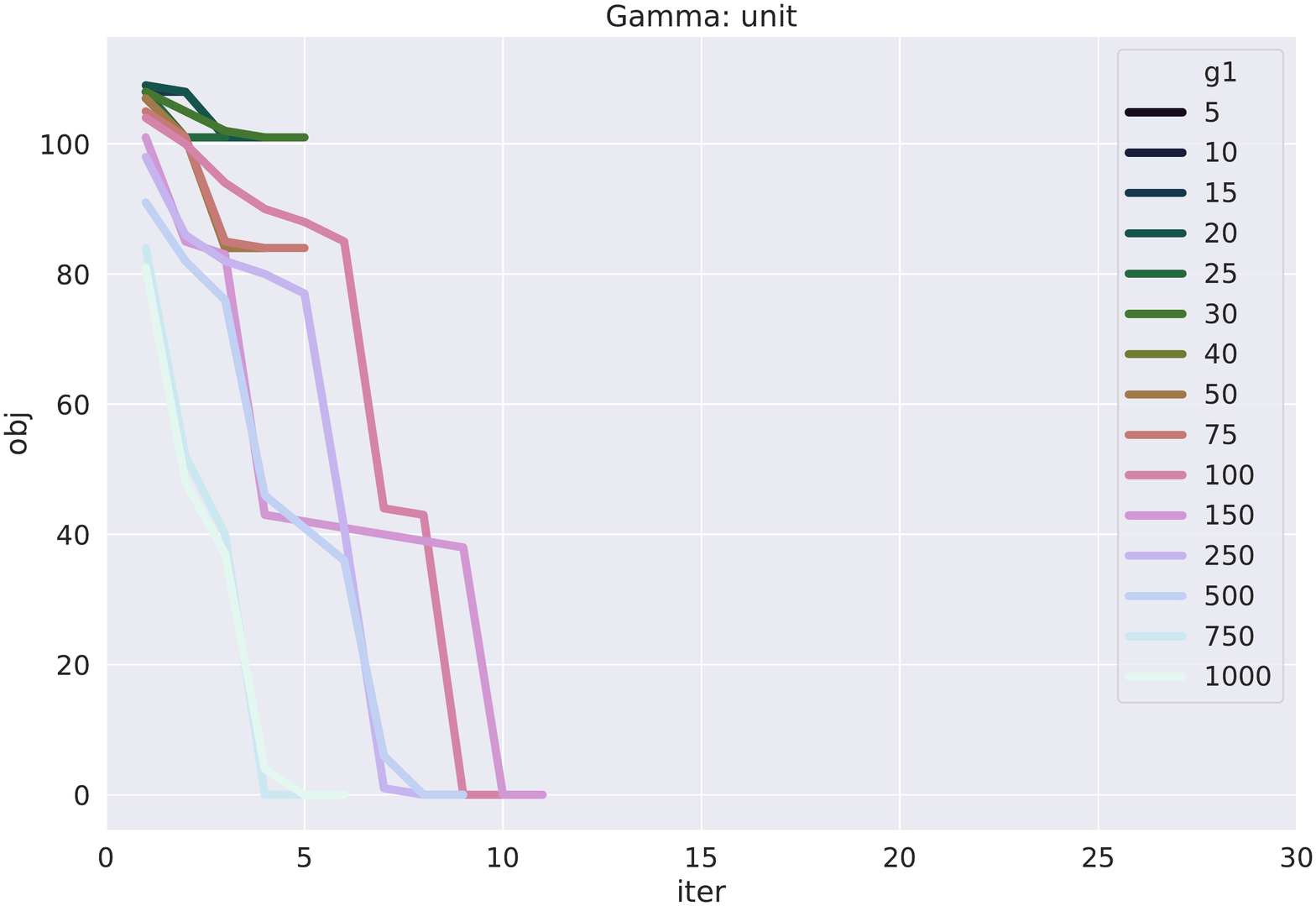}
  \caption{$Q||C_{\max}$, stacked plots, left-to-right $\Gamma_{\text{2-apx}}$, $\Gamma_{\text{5-apx}}$, $\Gamma_{\text{10-apx}}$, $\Gamma_{\text{best}}$and $\Gamma_{\text{any}}$.
  \label{fig:sched_single}
  }
\end{figure}

\begin{figure}[!h]
\centering
     \includegraphics[width=0.32\textwidth]{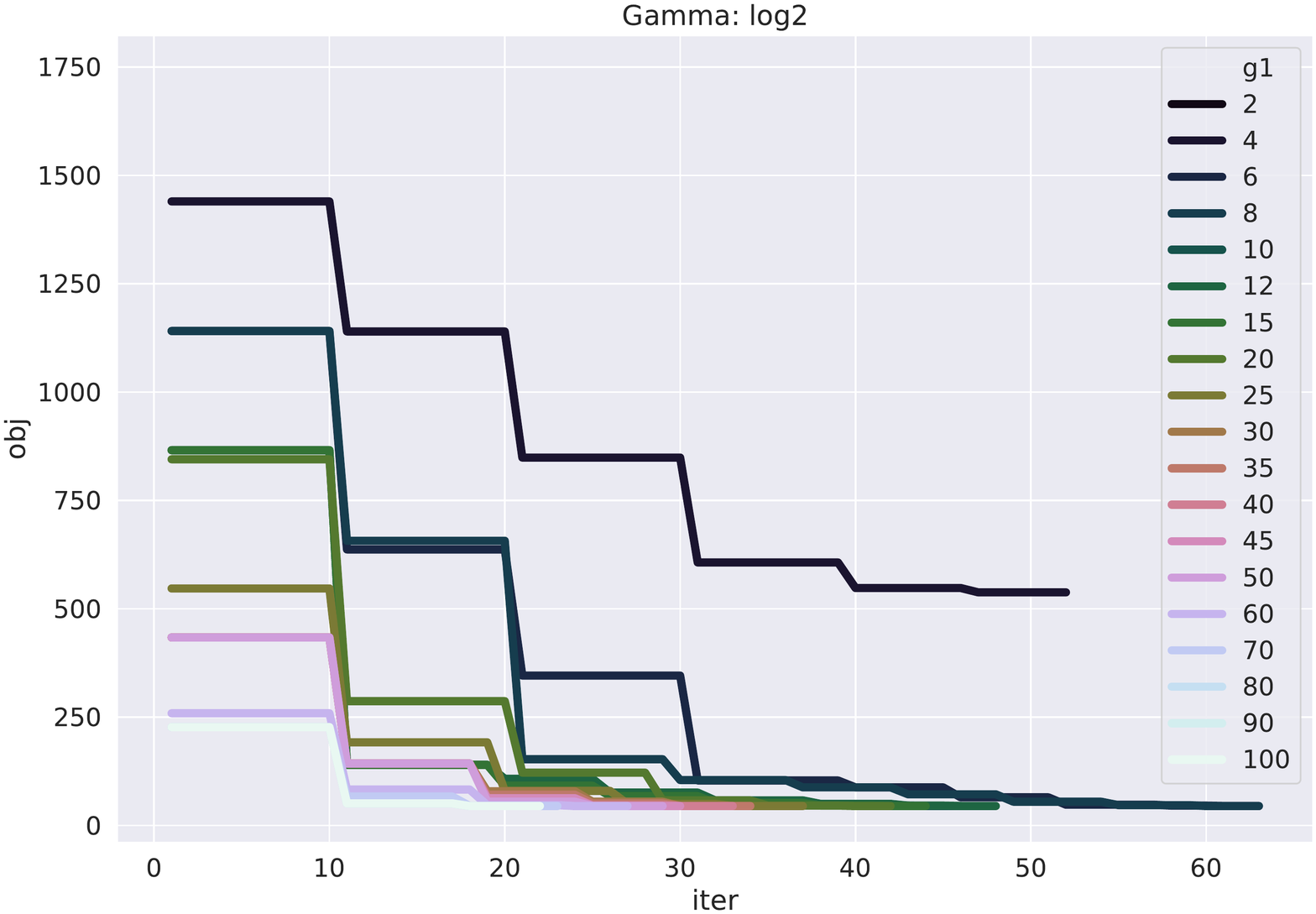}
     \includegraphics[width=0.32\textwidth]{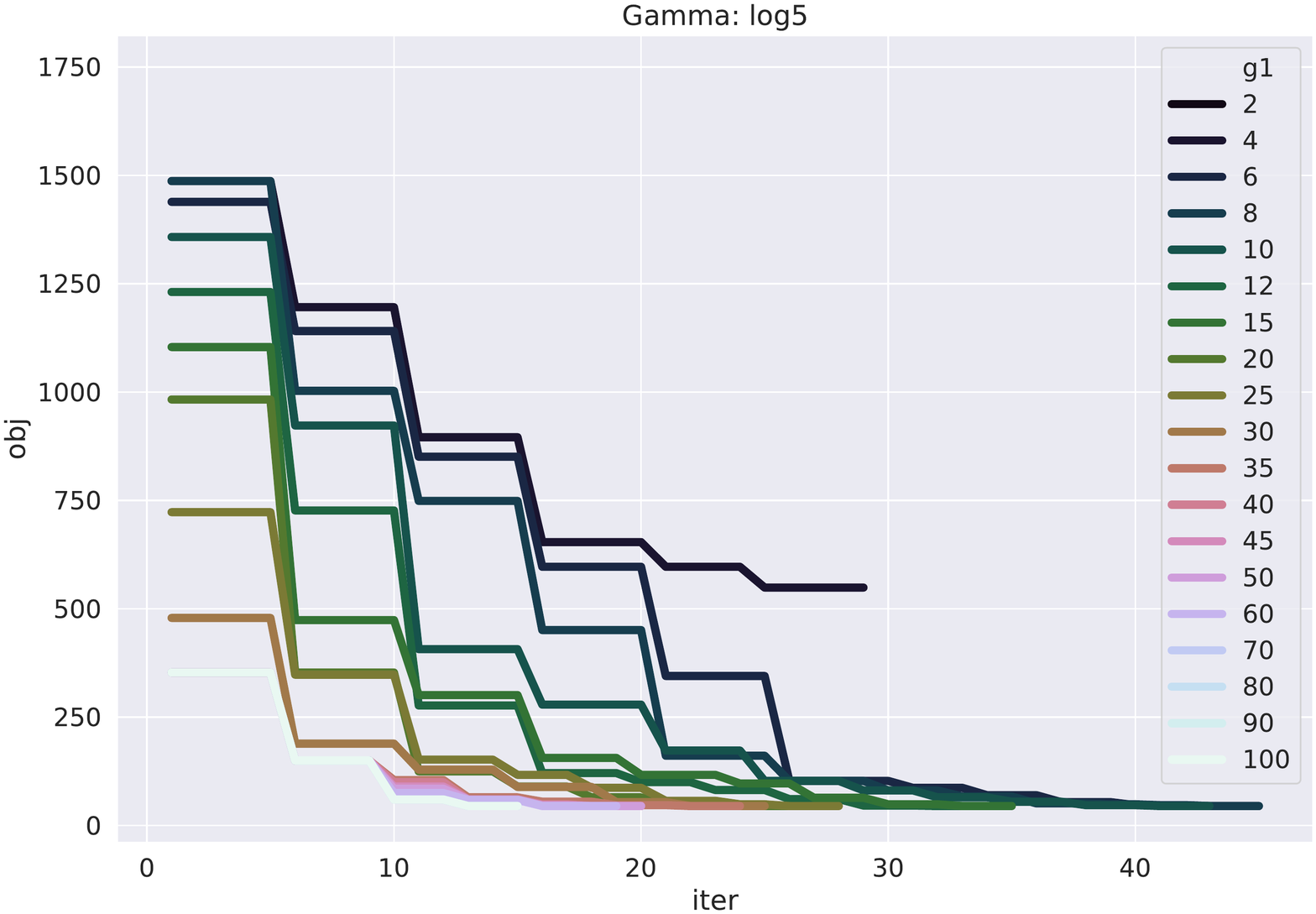}
     \includegraphics[width=0.33\textwidth]{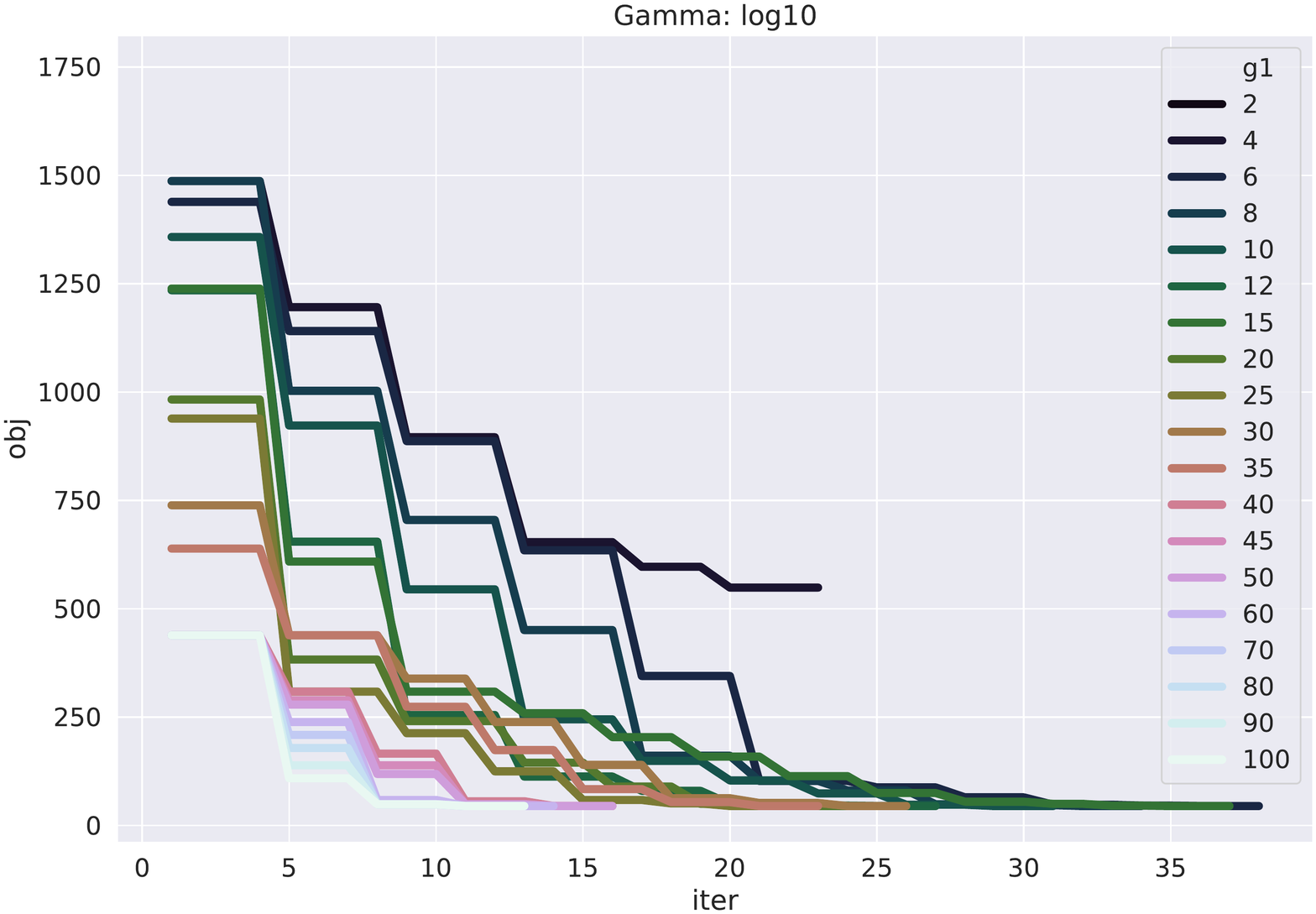}
     \includegraphics[width=0.49\textwidth]{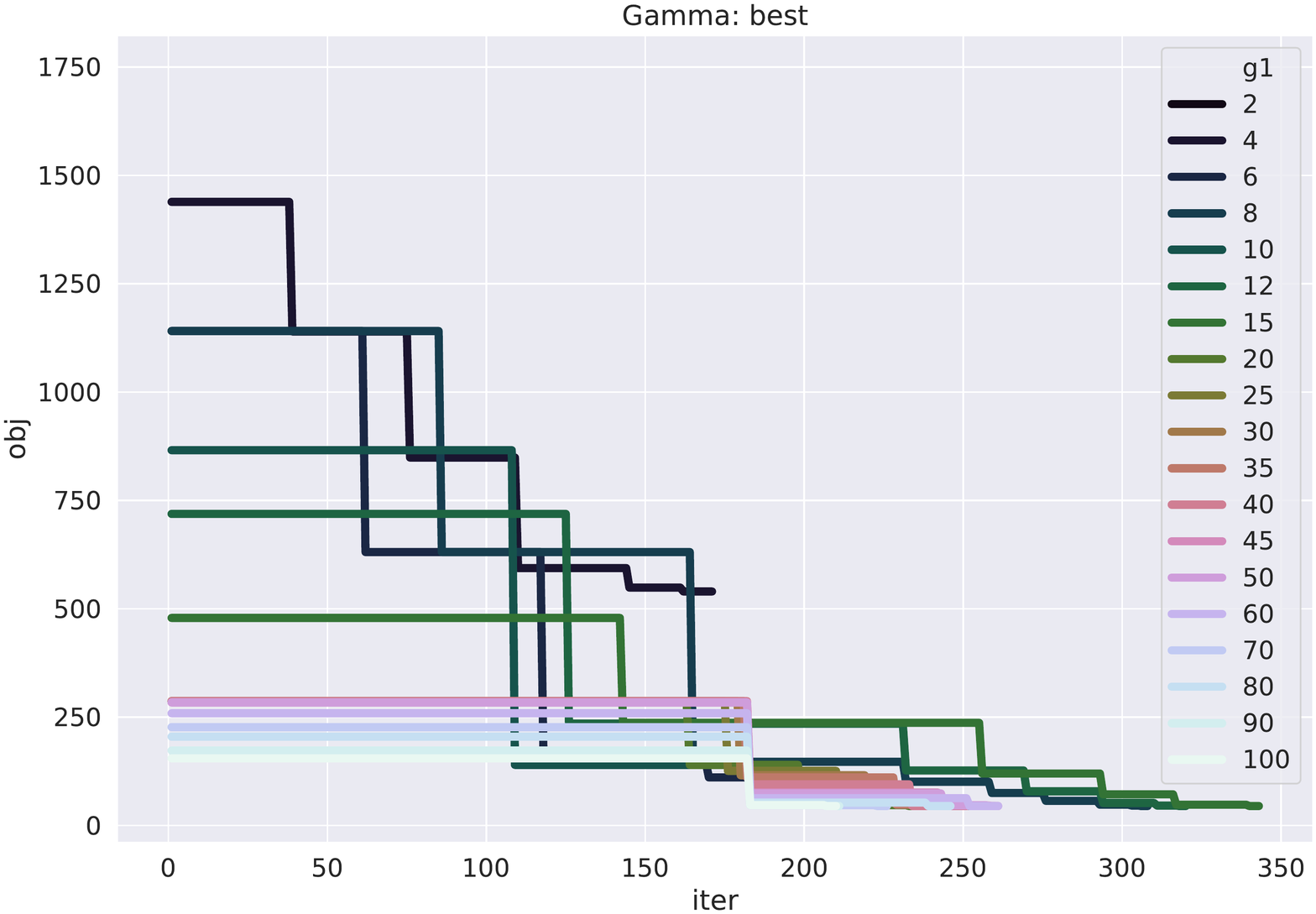}
     \includegraphics[width=0.49\textwidth]{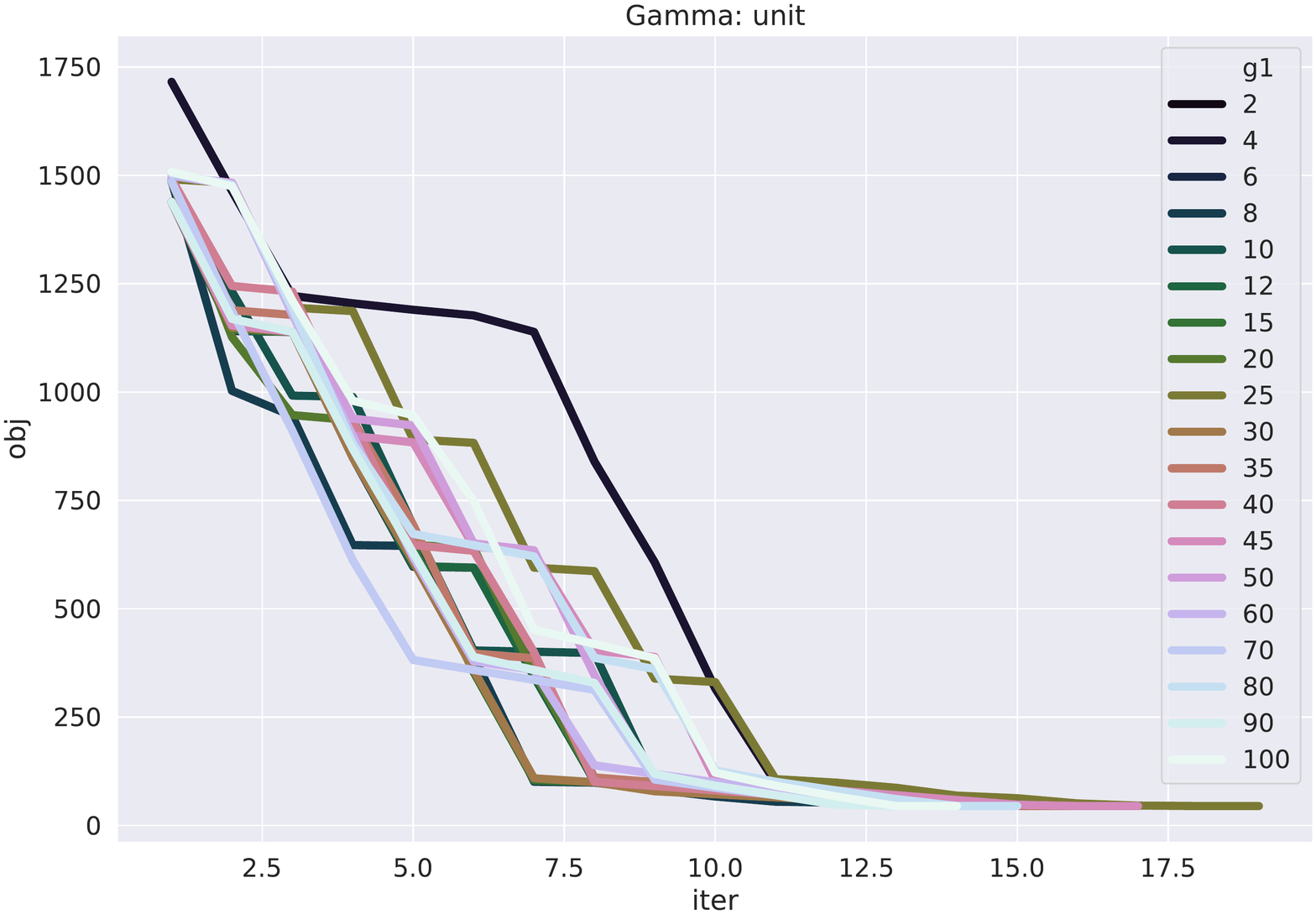}
  \caption{\textsc{Closest String}, stacked plots, left-to-right $\Gamma_{\text{2-apx}}$, $\Gamma_{\text{5-apx}}$, $\Gamma_{\text{10-apx}}$, $\Gamma_{\text{best}}$and $\Gamma_{\text{any}}$.
  \label{fig:cs_single}
  }
\end{figure}

\subsubsection*{Conclusions}
Our main takeaway regarding Question \#1 is that, while the theoretical upper bounds for $g_1(A)$ are huge, already small values of $\gc$ ($\gc >10$ for \textsc{Closest String} and $\gc > 150$ for \textsc{Makespan Minimization}) are sufficient for convergence to global optima.
We remark that, in the case of \textsc{Closest String}, this hints at the possibility that the maximum value of any \emph{feasible} augmenting step $\veg \in \G(A)$ is bounded by $k^{\Oh(1)}$ rather than $k^{\Oh(k)}$, which would imply an algorithm with runtime $k^{\Oh(k)} \log L$ while the currently best algorithm runs in time $k^{\Oh(k^2)} \log L$~\cite{KnopKM:2017esa}.

Regarding Question \#2, we see that $\Gamma_{\text{2-apx}}$ converges in a similar way as $\Gamma_{\text{best}}$ but is orders of magnitude cheaper to compute.
The ``any step'' augmentation strategy $\Gamma_{\text{any}}$ usually converges surprisingly quickly, but our results make it clear that its behavior is erratic and unpredictable.
Specifically, with augmentation strategies such as $\Gamma_{\text{2-apx}}$, increasing the parameter $\gc$ reliably leads to faster convergence, while for $\Gamma_{\text{any}}$ this is not the case.
Consequently, beyond some value of $\gc$ strategies such as $\Gamma_{\text{2-apx}}$ outperform $\Gamma_{\text{any}}$ in absolute numbers of iterations.

The detailed Figures~\ref{fig:sched_log2}-\ref{fig:sched_best} reveal that the step which is eventually taken is often found for relatively larger step-lengths $\lambda$; this explains why $\Gamma_{\text{5-apx}}$ outperforms $\Gamma_{\text{2-apx}}$ and is typically outperformed by $\Gamma_{\text{10-apx}}$, as $\Gamma_{c\text{-apx}}$ spends less time on short step-lengths with increasing $c$.

\subsection{Quantitative Evaluation}
In the second part of our evaluation, we relate several instance parameters to the two selected performance parameters.
The instance parameters of our interest are
\begin{itemize}
\item dimension $Nt$,
\item largest coefficient $\Delta$,
\item number of columns of the $E_1$ block, that is, $t$,
\item number of rows of the $E_1$ block, that is, $r$,
\item number of bricks $N$.
\end{itemize}
As we have noted for both problems, the matrix $E_2$ has only one row, so the parameter $s$ is always~$1$.
Moreover, for \textsc{Closest String} the parameters dimension, $N$, $t$, and $r$ are closely related, as the dimension is $Nt$ with $N=t^2$ and $t \leq r^r$.
To simplify matters, from now on we ran all tests with augmentation strategy $\Gamma_{\text{2-apx}}$.

Regarding performance parameters, we wish to study the \emph{optimality gap} which is simply the difference between the optimum obtained by the algorithm and the exact optimum.
Moreover, we wish to quantify the notion of a ``convergence rate'' in a normalized way to allow comparison across instances.
To this end, fix an instance and denote by $\operatorname{it}(\gc)$ the number of (inner) iterations taken by the algorithm to reach the optimum (and $+\infty$ if optimality gap is positive), let $\operatorname{it}_{\min} = \min_{\gc} \operatorname{it}(\gc)$, and finally let the \emph{convergence rate} be $c(\gc) = \operatorname{it}_{\min} / \operatorname{it}(\gc)$.
Thus $0 \leq c(\gc) \leq 1$ with $c(\gc) = 0$ if setting the tuning parameter to value $\gc$ does not make the algorithm find the optimum, and with larger values corresponding to faster convergence.

The testing batches were generated with the following parameters:
\begin{itemize}
\item For $Q||C_{\max}$, the command line was \texttt{`./nfold\_sched\_tester.sage --logdir 31012019 --machines 10 20 30 40 50 60 --count\_for\_each\_p 1 --slacks 0.6 0.7 --p\_s 6 7 8 9 10 11 12 13 14 15 16 17 18 19 20 --number\_job\_types 4 5 6 --gammas log2 --gc 5 10 15 20 25 30 40 50 75 100 150 250 500 750 1000`}.
\item For \textsc{Closest String}, the command line was \texttt{./nfold\_sched\_tester.sage --instance\_type cs --logdir cs\_test --milp\_timelimit 300 --augip\_timelimit 300} (i.e., all relevant parameters left to defaults).
\end{itemize}

\subsubsection*{Plots}
We visualize the relationships as follows: each plot in Figures~\ref{fig:sched_hmaps} and~\ref{fig:cs_hmaps} is a heatmap whose columns are increasing values of $\gc$, rows are increasing values of $\Delta$ or dimension (for $Q||C_{\max}$), and cells are values of optimality gap or convergence rate.
The color scheme is such that darker shades correspond to worse behavior, be it larger optimality gap or smaller convergence rate.

\begin{figure}[!h]
\centering
\includegraphics[width=0.49\textwidth]{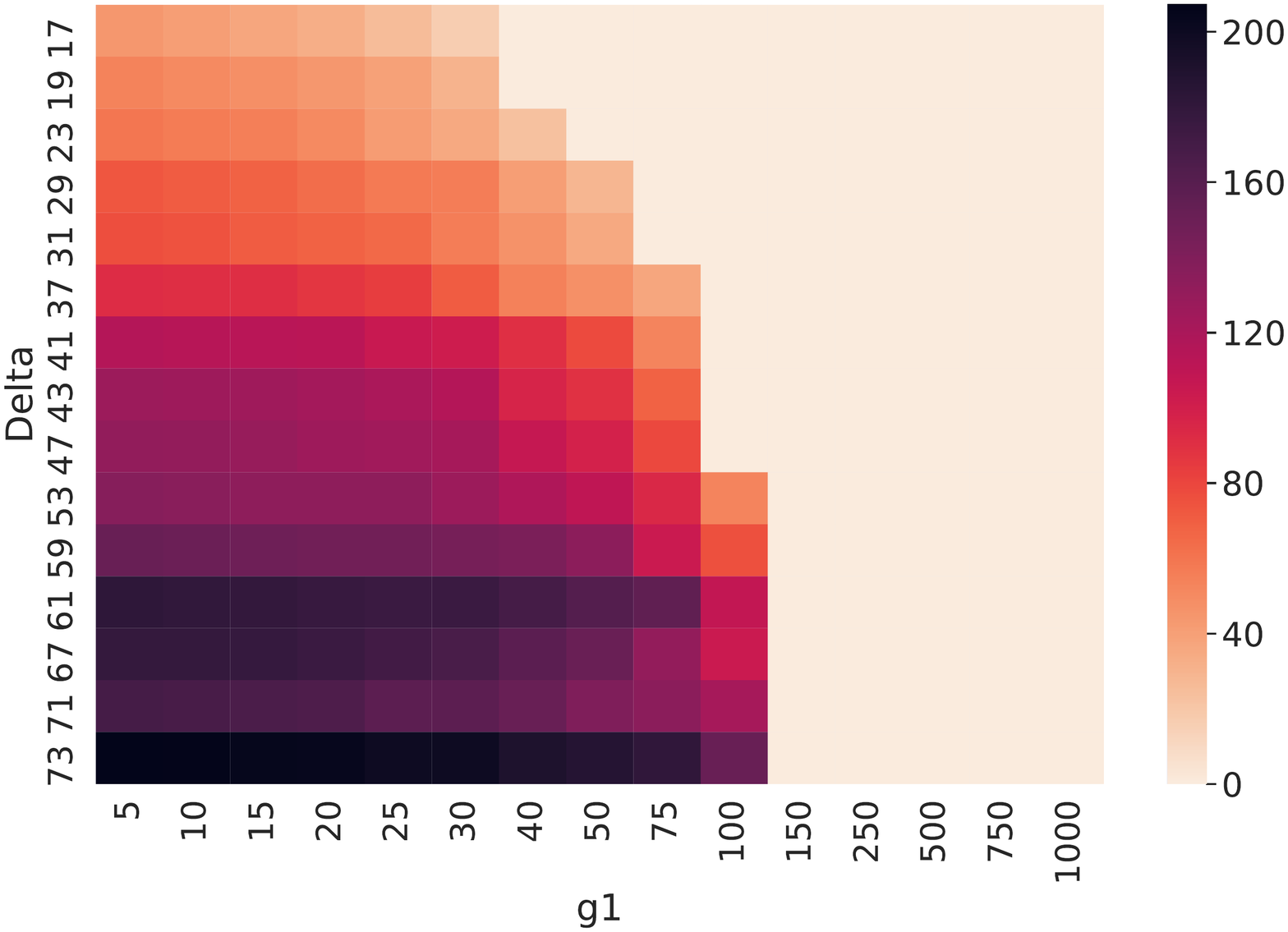}
\includegraphics[width=0.49\textwidth]{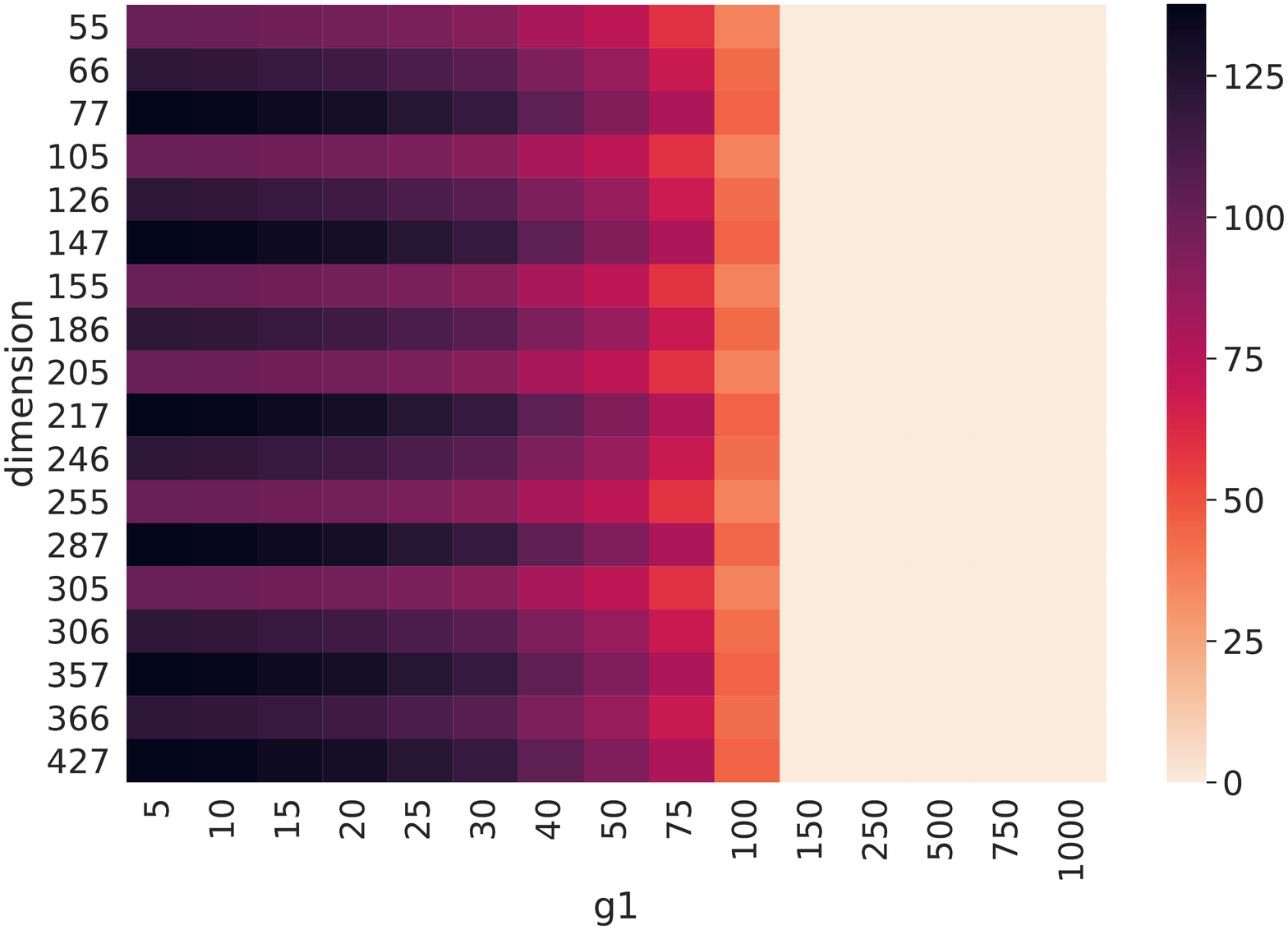}\\
\includegraphics[width=0.49\textwidth]{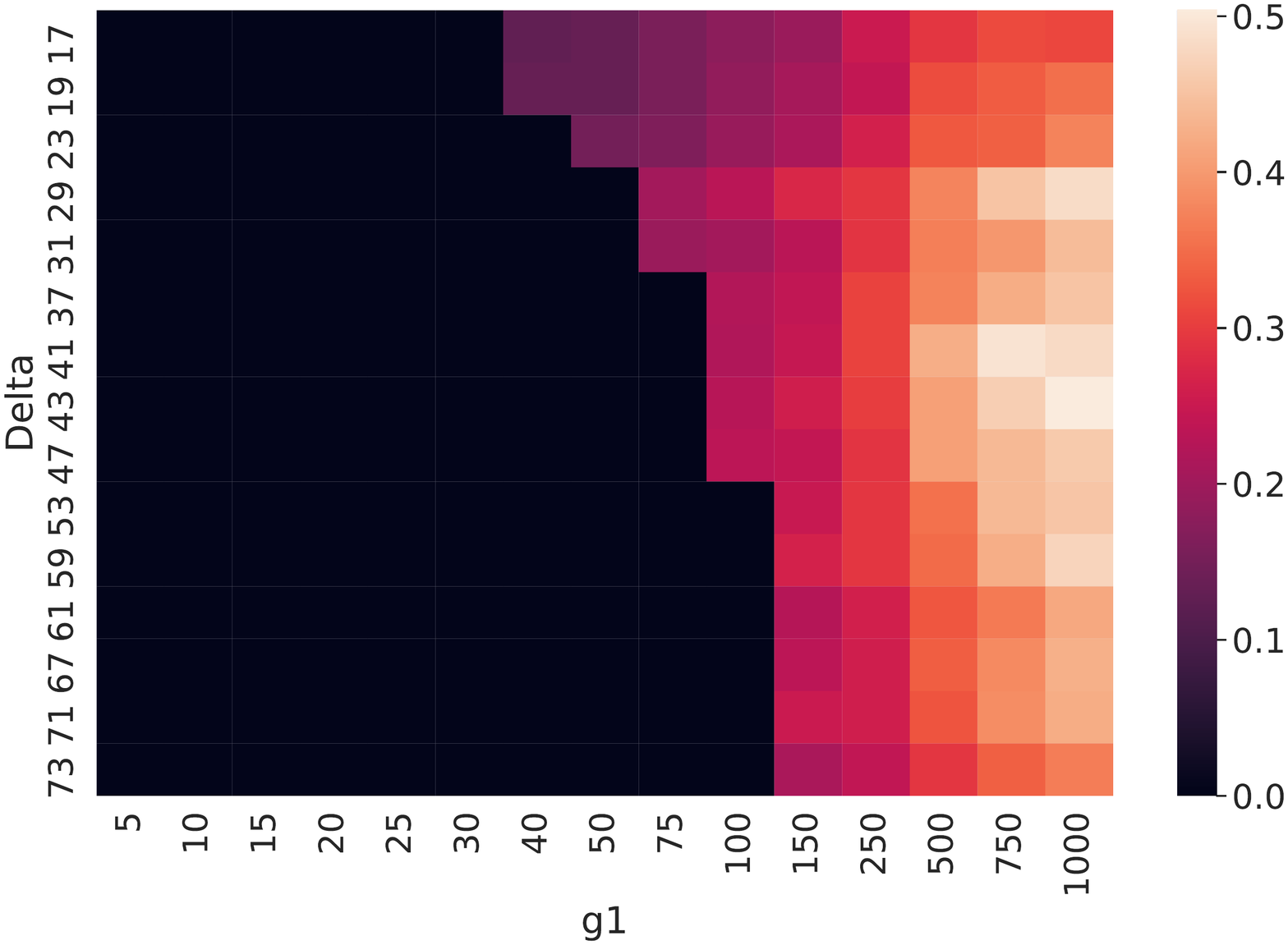}
\includegraphics[width=0.49\textwidth]{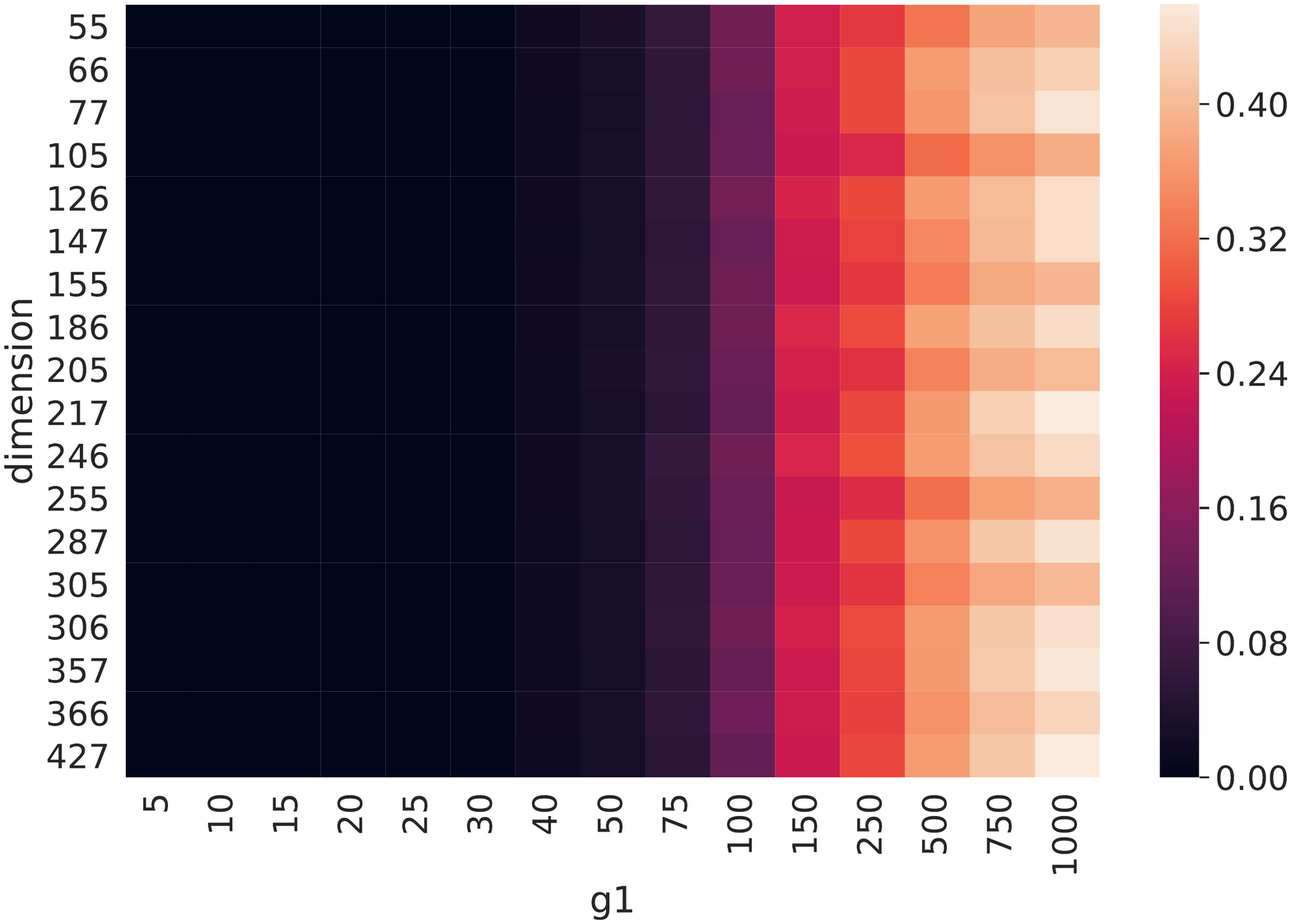}
  \caption{$Q||C_{\max}$, heatmaps, columns are values of $\gc$, and, left-to-right, rows are $\Delta$, dimension, $\Delta$, dimension, and cells are gap, gap, convergence, convergence., respectively.
  \label{fig:sched_hmaps}
  }
\end{figure}

\begin{figure}[!h]
\centering
\includegraphics[width=0.49\textwidth]{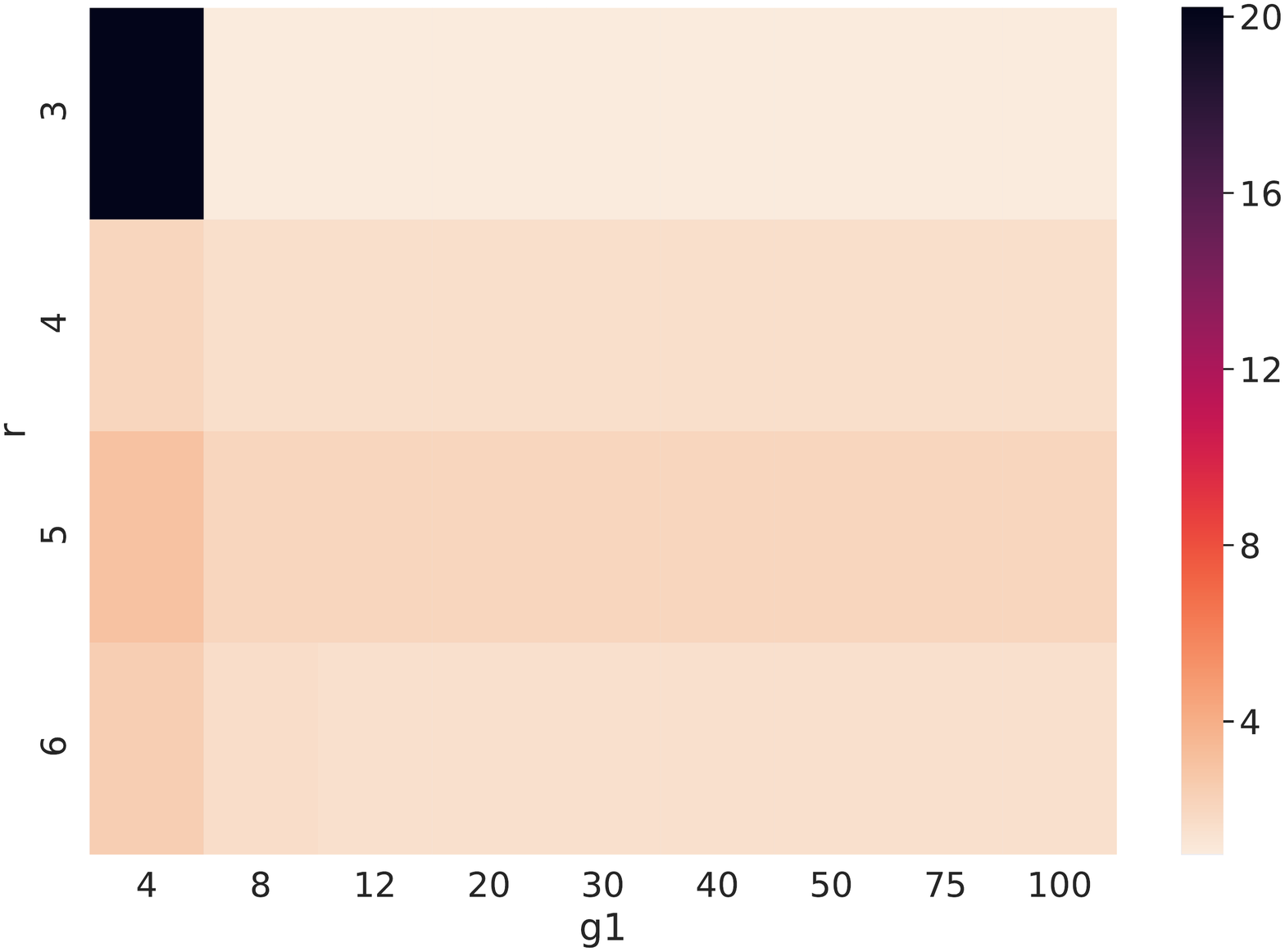}
\includegraphics[width=0.50\textwidth]{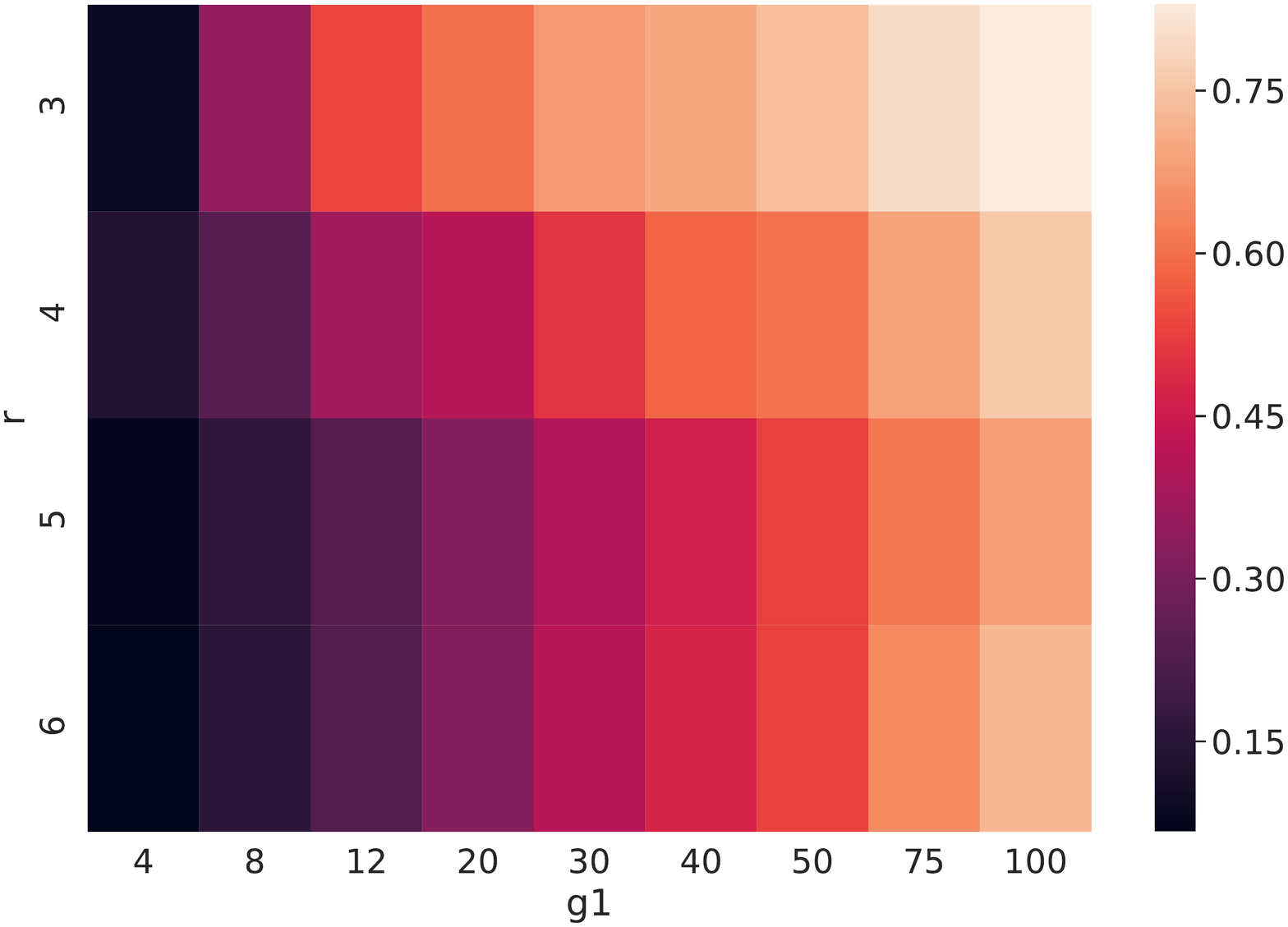}
  \caption{\textsc{Closest String}, heatmaps, columns are values of $\gc$, rows are values of $r$, cells are, left-to-right, gap and convergence.
  \label{fig:cs_hmaps}
  }
\end{figure}


%

\subsubsection*{Conclusions}
Our results, now measured across many instances, confirm our previous hypotheses: increasing values of $\gc$ lead to decreasing optimality gaps and improving convergence rates.
Moreover, the effect seems to correlate more with $\Delta$ than the dimension $Nt$.
To see this observe in particular Figure~\ref{fig:sched_hmaps} whose rows and columns look similar, indicating relatively small correlation with the dimension as compared with $\Delta$.

This corresponds to the theoretical observation that the ``true'' value of $g_1(A)$ is independent of $t$ and $N$ but depending on $\Delta$ and $r$.
(Note that for $Q||C_{\max}$ we have $t=r$ so the parameter $t$ is expected to have an effect, however, our tested values cover possibly a too narrow range.)

\subsection{Towards Practical Applications}
So far we have been interested in parameters ``internal'' to the implemented algorithm.
In particular, we have disregarded actual \emph{time} taken by the computation and any analysis of potential bottlenecks of the algorithm.
The relevant time parameters which we study now are the following:
\begin{itemize}
\item total time needed to run Algorithm~\ref{alg:New}, denoted \texttt{total},
\item time required to initialize the MILP model of~\eqref{AugIP}, denoted~\texttt{augip init}
\item time consumed by solving~\eqref{AugIP} excluding initialization, denoted \texttt{augip total}, and,
\item time required to construct the MILP model of the whole instance~\eqref{IP} and solve it using Gurobi, denoted \texttt{gurobi construct \& solve}.
\end{itemize}

Our initial observation during preliminary experiments was that the total required time grows significantly with increasing dimension.
Thus our goal was to determine potential instance parameters such as dimension or $\Delta$ which make the instance hard for Gurobi, with the hope that for such instances a good implementation of a parameterized $N$-fold IP algorithm would outperform Gurobi.
However, a closer examination has revealed that the observed growth is caused by increasing time taken by the model construction phase (\texttt{aug init} and the ``construct'' part of \texttt{gurobi construct \& solve}).

\subsubsection*{Plots}
We present our findings in two types of plots.
The first one (Figures~\ref{fig:sched_time_lineplots} and~\ref{fig:cs_time_lineplots}) is a line plot whose $y$ axis is time and $x$ axis is one of dimensions, $r$, and $\Delta$ (only for $Q||C_{\max}$), with individual lines corresponding to the different time parameters above.
Semi-transparent bands around lines correspond to 95\% confidence intervals.

\begin{figure}[!h]
\centering
\includegraphics[width=0.32\textwidth]{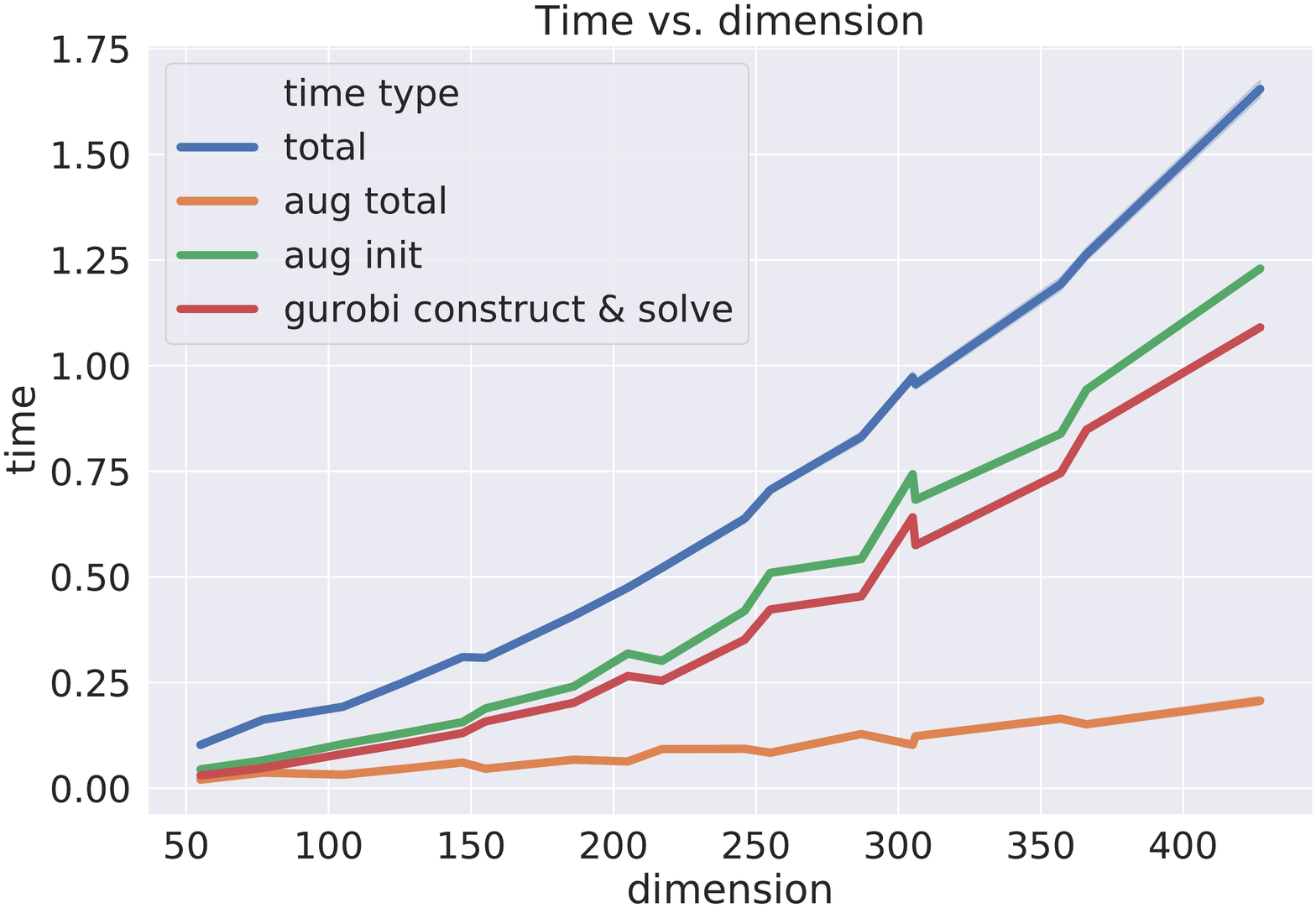}
\includegraphics[width=0.32\textwidth]{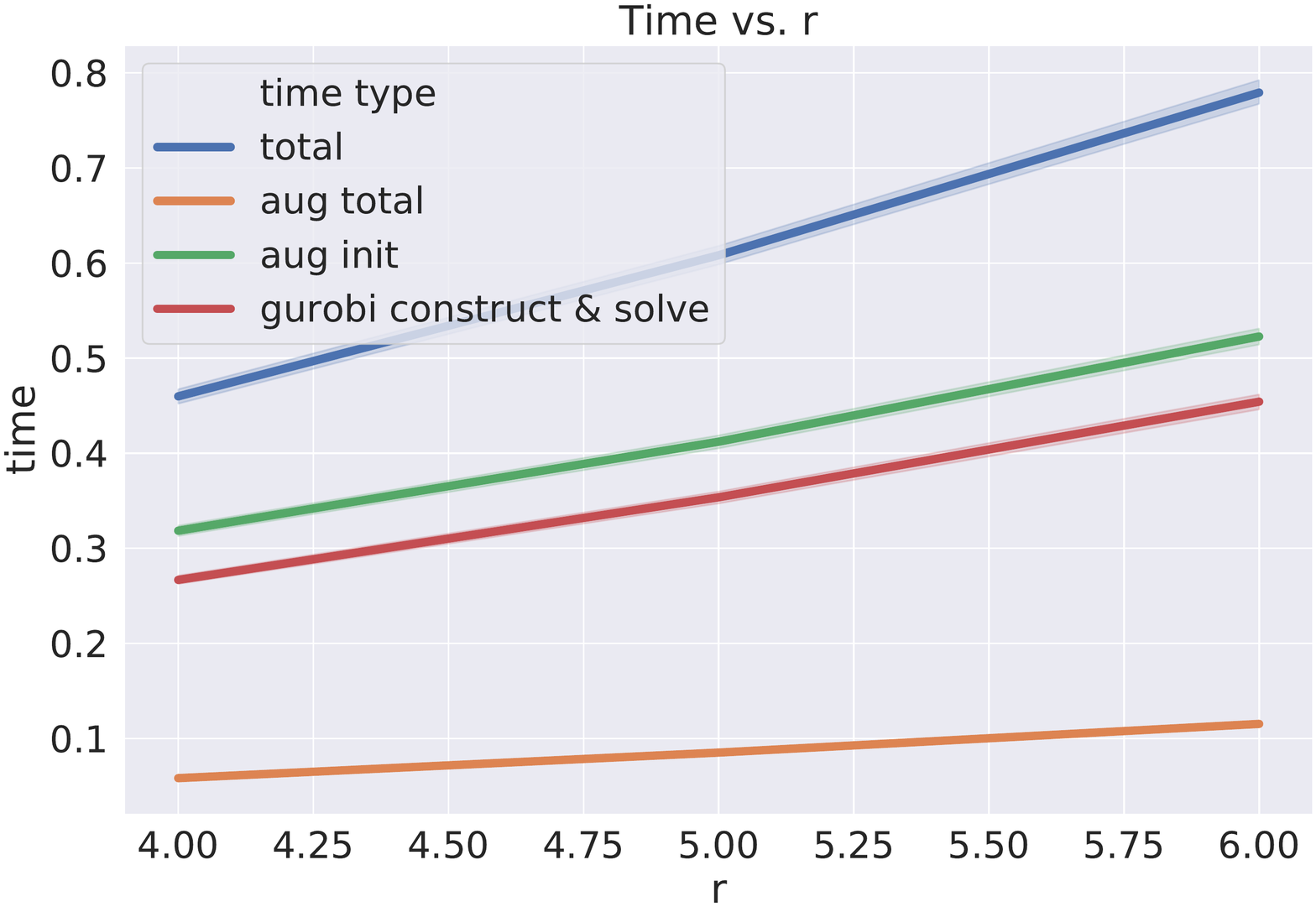}
\includegraphics[width=0.32\textwidth]{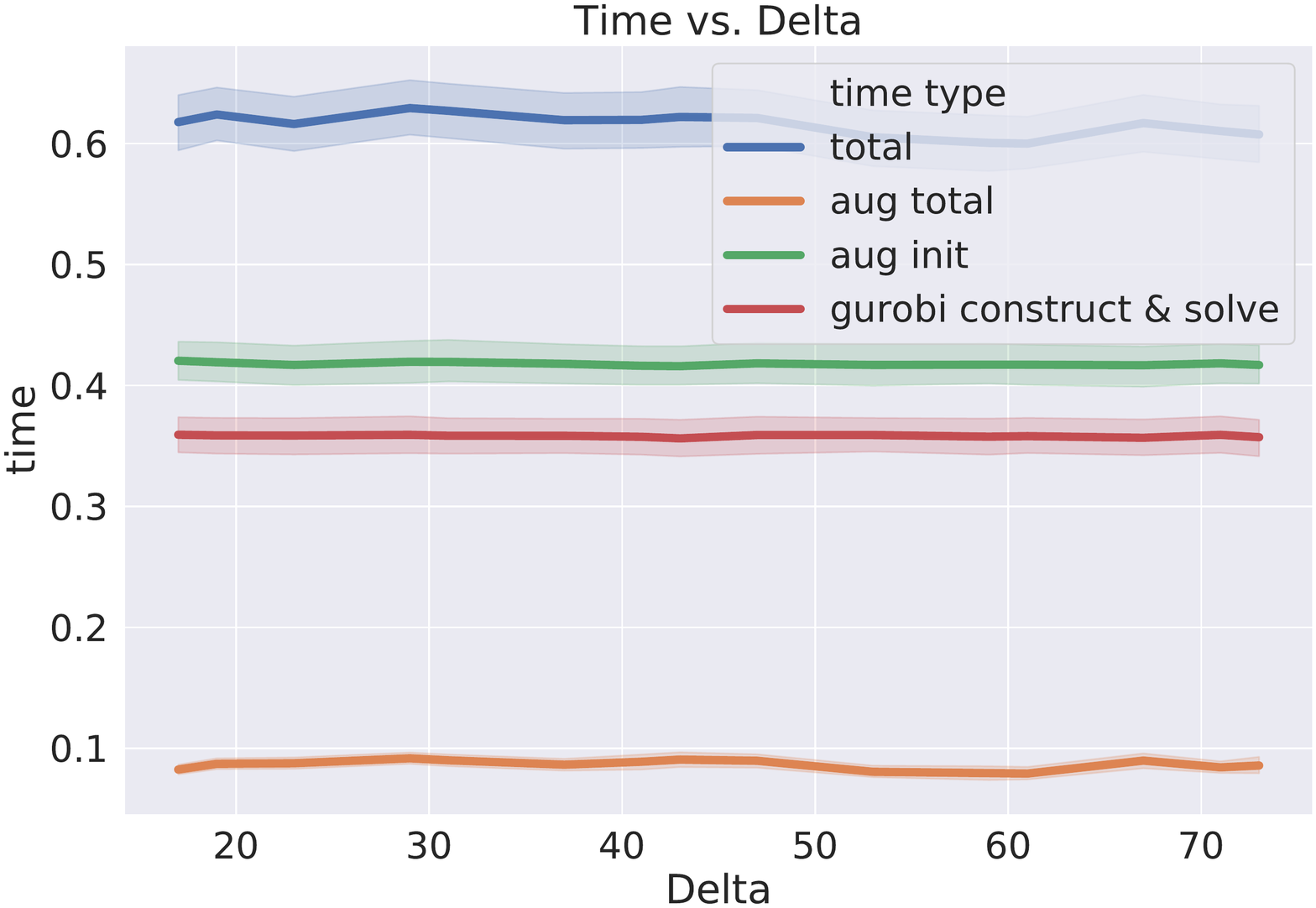}
  \caption{$Q||C_{\max}$, line plots, $y$ axis is time, $x$ axis is, left-to-right, dimension, $r$, and $\Delta$, respectively.
  \label{fig:sched_time_lineplots}
  }
\end{figure}

\begin{figure}[!h]
\centering
\includegraphics[width=0.49\textwidth]{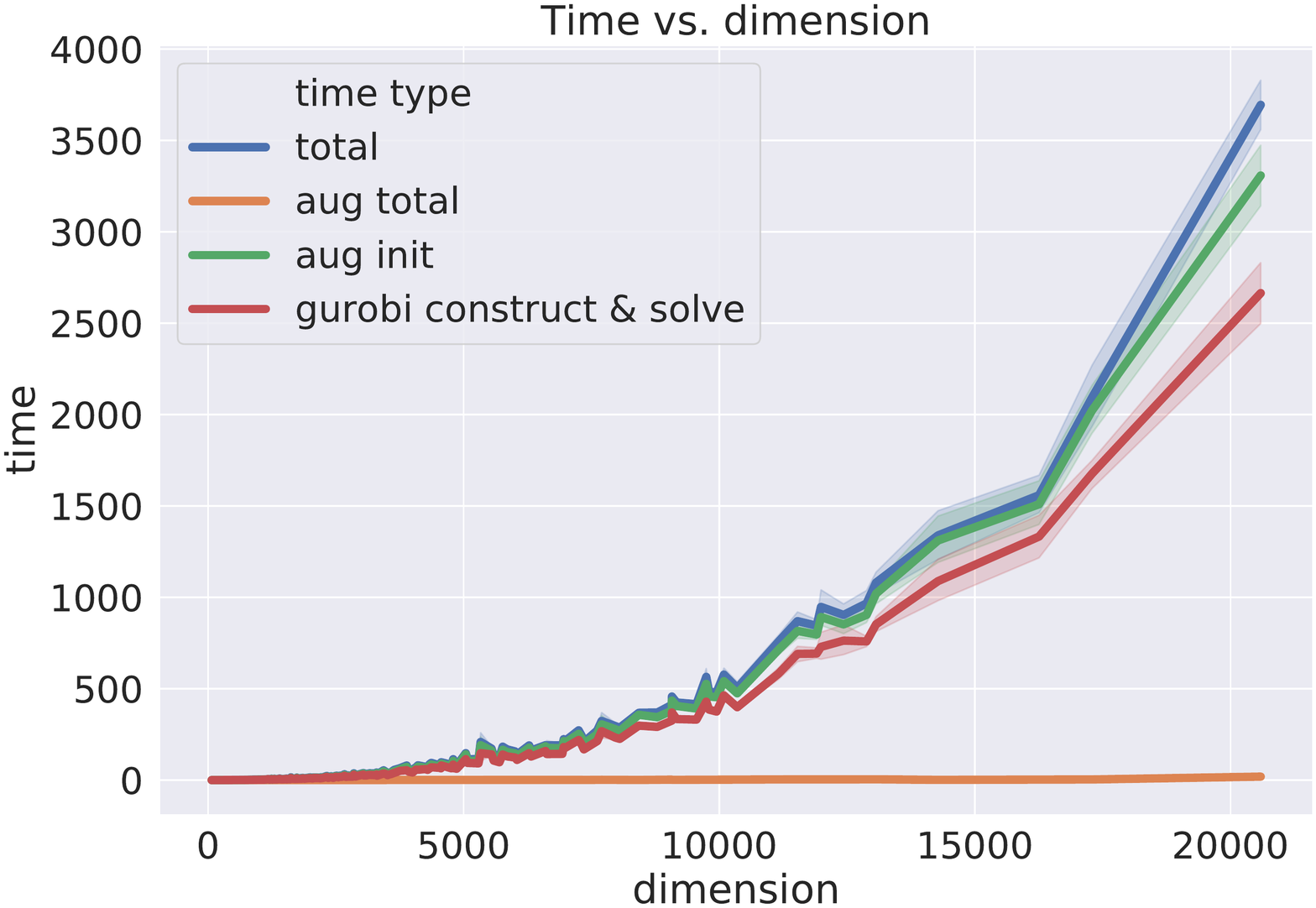}
\includegraphics[width=0.50\textwidth]{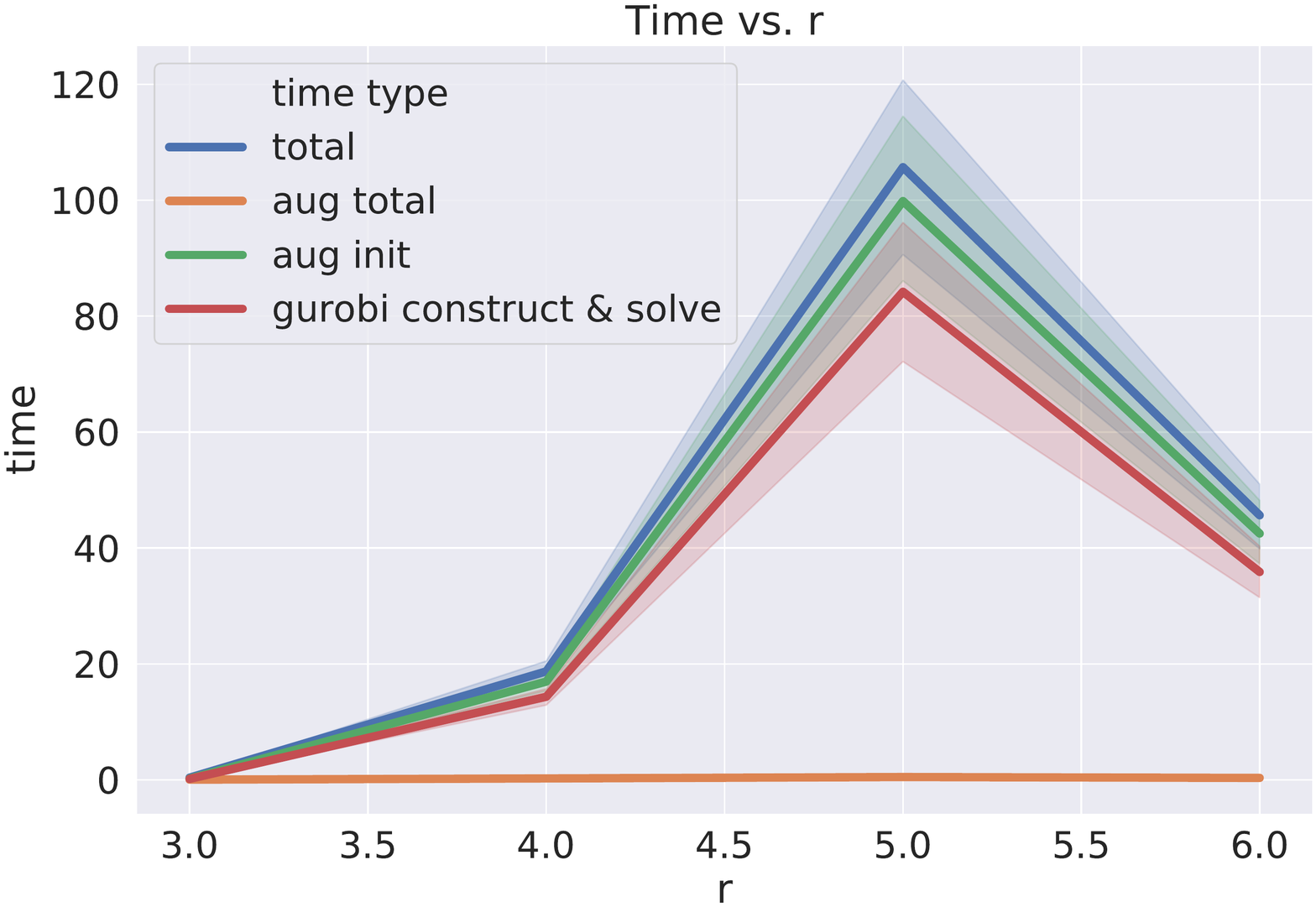}
  \caption{\textsc{Closest String}, line plots, $y$ axis is time, $x$ axis is, left-to-right, dimension and $r$.
  \label{fig:cs_time_lineplots}
  }
\end{figure}

The second type (Figure~\ref{fig:sched_time_hmaps}) constructed only for $Q||C_{\max}$ shows the individual time parameters with respect to dimension and $\Delta$, in the form of heatmaps.

\begin{figure}[!h]
\centering
\includegraphics[width=0.49\textwidth]{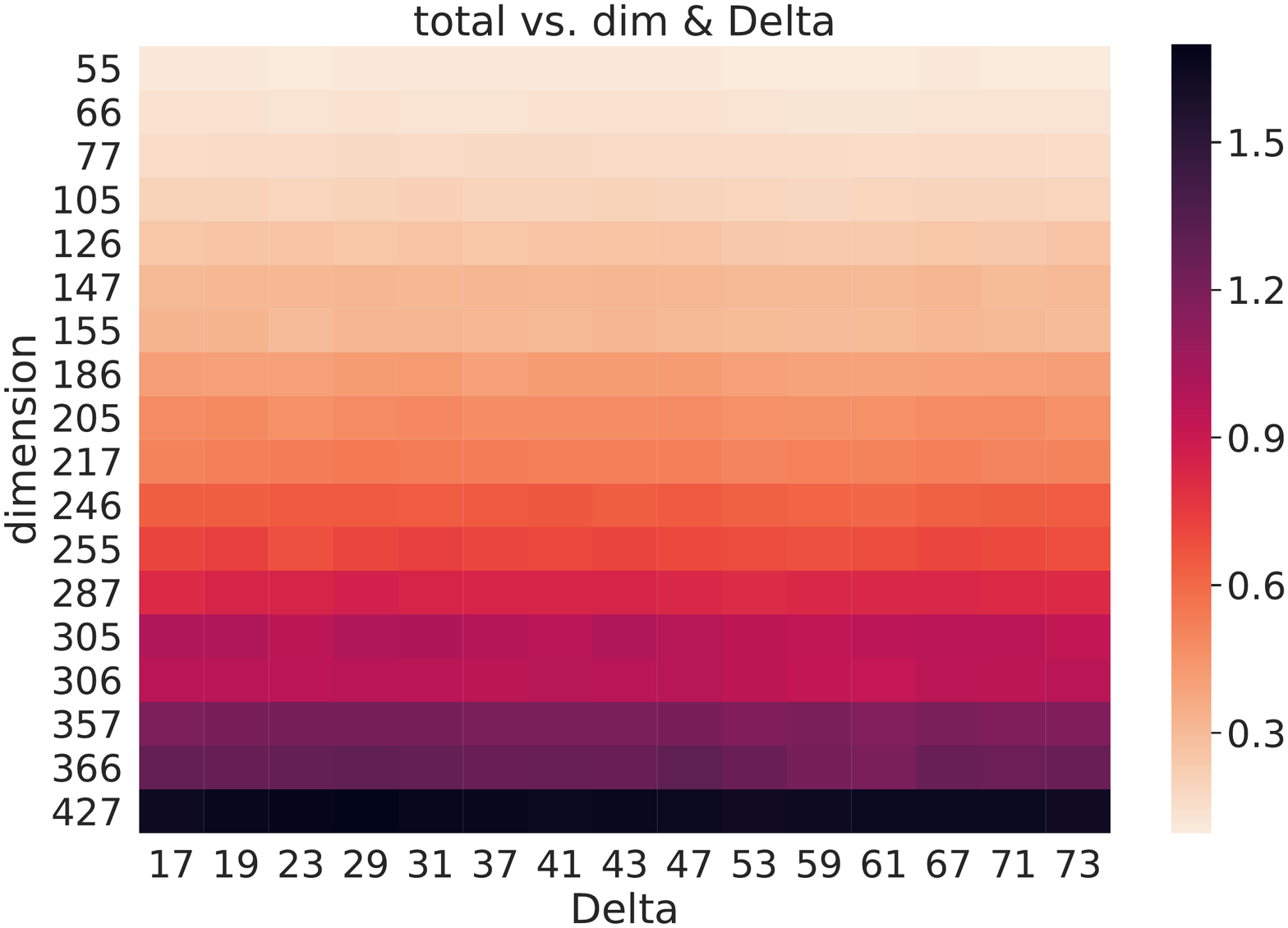}
\includegraphics[width=0.49\textwidth]{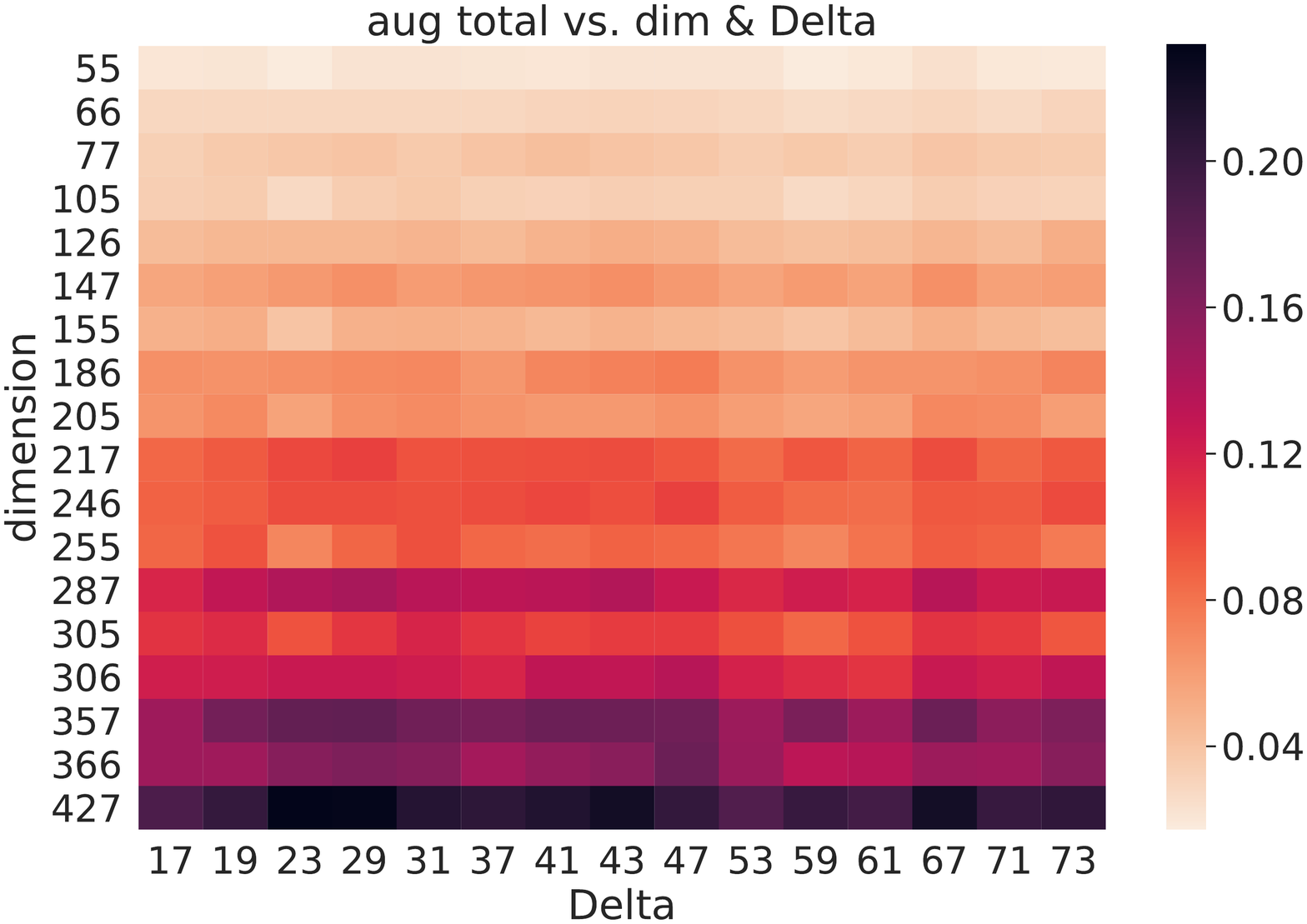}\\
\includegraphics[width=0.49\textwidth]{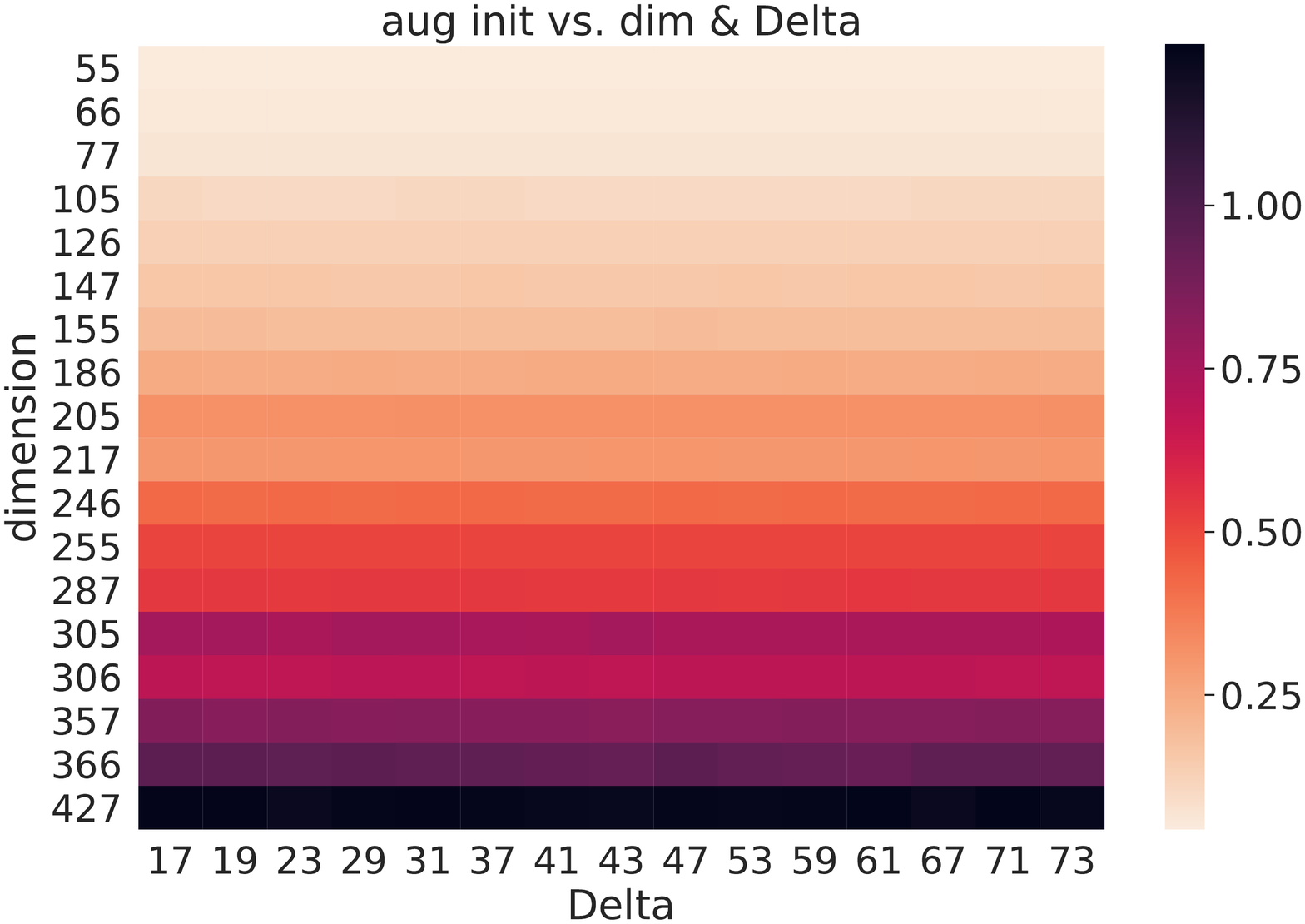}
\includegraphics[width=0.49\textwidth]{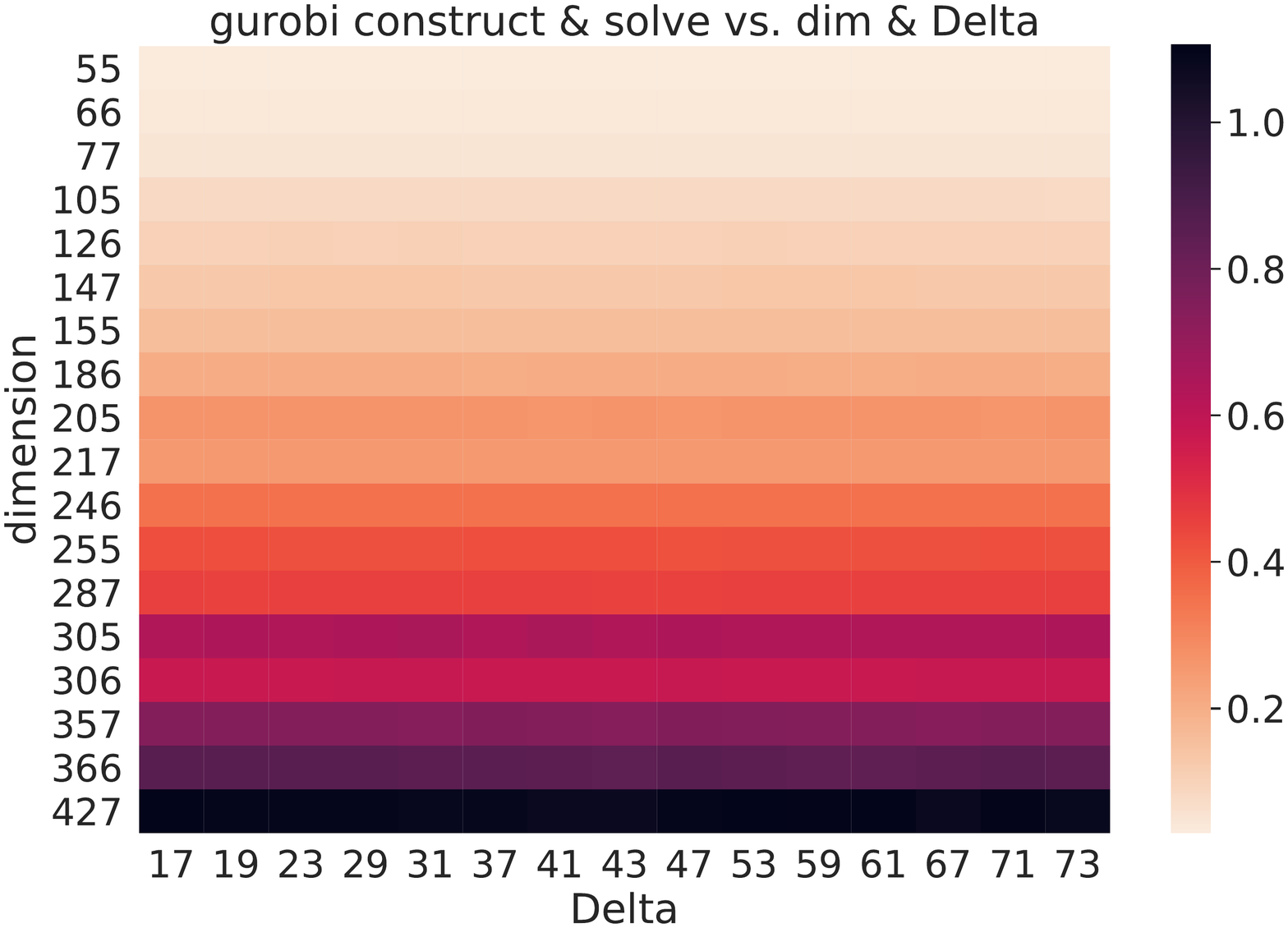}
  \caption{$Q||C_{\max}$, heatmaps, columns are $\Delta$, rows are dimension, cells are time parameters, left-to-right, \texttt{total}, \texttt{aug total}, \texttt{aug init}, and \texttt{gurobi construct \& solve}.
  \label{fig:sched_time_hmaps}
  }
\end{figure}

\subsubsection*{Conclusions}
Unfortunately, we conclude that, at least in the case of our formulations of \textsc{Closest String} and $Q||C_{\max}$, neither increasing $\Delta$ nor dimension create an obstacle for Gurobi itself.
Instead, the bottleneck lies in the overhead of SageMath and Python data structures.

\section{Outlook}
We have initiated an experimental investigation of a certain subclass of ILP with a block structured constraint matrix.
Our results show that, as theory suggests, for such ILPs a primal algorithm always augmenting with steps of small $\ell_1$ norm converges quickly.
We close with a few interesting research directions.

First, in theory, the special structure of~\eqref{AugIP} (in particular, an $\ell_1$-norm bound on its solution) as compared with~\eqref{IP} means that~\eqref{AugIP} can be solved faster than~\eqref{IP}.
However, in practice, this seems to have little to no effect.
Thus we ask: is there a way to tune generic MILP solvers to solve~\eqref{AugIP} significantly faster than~\eqref{IP}?

Second, what is the behavior of our algorithm on instances \emph{other} than $N$-fold IP?
For example, how large does $\gc$ have to be in order to attain the optimum quickly for standard benchmark instances, e.g. MIPLIB~\cite{miplib}?

Third, the approach of Koutecký et al.~\cite{PSP} suggests that a key property for the efficient solvability of~\eqref{AugIP} is a certain ``sparsity'' and ``shallowness'' (formally captured by the graph parameter \emph{tree-depth}) of graphs related to the constraint matrix.
Thus we ask what are ``natural'' instances with small tree-depth, and what is ``typical'' tree-depth of instances used in practice.

\bibliography{experiments}

\end{document}